\documentclass[twoside,journal]{IEEEtran}
\usepackage{bbm}
\usepackage{tipa}
\usepackage[noend]{algpseudocode}
\usepackage{algorithmicx,algorithm}
\usepackage{float}
\usepackage{lipsum}

\usepackage[dvips]{graphicx}
\usepackage{amsfonts}
\usepackage{amssymb}
\usepackage{ntheorem}
\usepackage{verbatim}
\usepackage{subfigure}
\usepackage{balance}
\usepackage{booktabs}
\usepackage{fancybox}
\usepackage{bm}
\usepackage{pifont}
\usepackage{extarrows}
\usepackage{cite}
\usepackage{color}
\newcounter{TempEqCnt}
\newtheorem{theorem}{Theorem}

\newtheorem{proposition}{Proposition}
\newtheorem*{proof}{\it{Proof:}}

\begin{document}
\bibliographystyle{IEEEtran}

\title{Two-phase Unsourced Random Access in Massive MIMO: Performance Analysis and Approximate Message Passing Decoder
\thanks{\scriptsize
The work was supported in part by the
National Natural Science Foundation of China under Grant 62171364 and 61941118. (\emph{Corresponding author: Hui-Ming
Wang.})}
\author{Jia-Cheng Jiang and Hui-Ming Wang, \emph{Senior Member}, \emph{IEEE}\hspace{0.02in}}
\thanks{\scriptsize
The authors are with the School of Information and Communications Engineering, Xi'an Jiaotong
University, Xi'an, 710049, Shaanxi, P. R. China, and also with the Ministry
of Education Key Lab for Intelligent Networks and Network Security, Xi'an, 710049, Shaanxi, P. R. China. (e-mail: j1143484496b@stu.xjtu.edu.cn; xjbswhm@gmail.com.)}
}
\maketitle

\begin{abstract}
 In this paper, we design a novel two-phase unsourced random access (URA) scheme in massive multiple input multiple output (MIMO). In the first phase, we collect a sequence of information bits to jointly acquire the user channel state information (CSI) and the associated information bits. In the second phase, the residual information bits of all the users are partitioned into sub-blocks with a very short length to exhibit a higher spectral efficiency and a lower computational complexity than the existing transmission schemes in massive MIMO URA. By using the acquired CSI in the first phase, the sub-block recovery in the second phase is cast as a compressed sensing (CS) problem. From the perspective of the statistical physics, we provide a theoretical framework for our proposed URA scheme to analyze the induced problem based on the replica method.
 The analytical results show that the performance metrics of our URA scheme can be linked to the system parameters by a single-valued free entropy function. An AMP-based recovery algorithm is designed to achieve the performance indicated by the proposed theoretical framework. Simulations verify that our scheme outperforms the most recent counterparts.

\end{abstract}
\begin{IEEEkeywords}
Unsourced Random Access, two-phase, massive MIMO, massive machine-type communication, replica method, statistical physics, approximate message passing decoder.
\end{IEEEkeywords}
\section{Introduction}
Massive machine-type communication (mMTC) is one of typical application scenarios in the fifth-generation (5G) and beyond 5G (B5G) wireless networks. The core mission of mMTC is to provide cellular connectivity to millions of low-rate machine-type devices for Internet of Things (IoT) applications \cite{ToMa2016, IoT5Gera2016}. Different from the traditional human-type communication (HTC), aiming to achieve high data rates with a large packet size, the machine-type communication (MTC) applications are typical uplink-driven with packet sizes as small as a few bits \cite{BockelmannMassive2016}. Another characteristic in mMTC is that the traffic pattern is typically sporadic with only a fraction of total user devices being active \cite{SenelGrant-Free2018, LiuMassive12018}.

Such features make the grant-free access control more favorable for mMTC than the grant-based access control. The grant-free access control requires a very low control overhead without the requirements of additional control signaling exchanges to facilitate the granting of resources. At the affordable cost of increased base station (BS) complexity, the collisions resolution mechanism can be designed by taking the sporadic traffic into consideration. One typical scenario for massive connectivity is that the BS first jointly detects the active users and estimates the corresponding channels, and then the identified active users are scheduled to transmit their messages \cite{LiuMassive12018}. Such a scenario is called \emph{sourced random access}, since the BS is required to identify the active users. The core mission of the sourced random access is to resolve the joint active detection and channel estimation (JADCE) problem. It is shown that the problem can be formulated as a sparse recovery problem, and some compressed sensing (CS) technologies for JADCE have recently been reported \cite{jiang2021, jiang20212}.


There is another line of work in the applications of mMTC, where the active user devices wish to transmit information bit sequences to the BS in an uncoordinated fashion, and the BS is only interested in the transmitted sequences but not user identifications. Such a manner is so-called \emph{unsourced random access} (URA). All the active users employ a same codebook and the decoder in the BS outputs an unordered list of the information bits. The URA framework was introduced by Polyanskiy in \cite{Polypersp}, which provided an existence bound for URA, using a random Gaussian codebook with maximum likelihood-decoding at the BS. Indeed, based on the CS-coding scheme, URA can be established as a support recovery in a very high dimension of the context of CS. However, the computational overhead of such the CS problem increases exponentially with the length of information bit sequences, and the excessive size precludes the straightforward application of existing solutions. As a consequence, several works have considered a divide-and-conquer approach to split the information sequences of active user devices into several sub-blocks with each of them processed separately, and finally, the information sub-blocks correspond to one original sequence are spliced together \cite{Amalladinne2020, Enhance2020, 2020Vamsi}. This approach was first introduced in \cite{Amalladinne2020}, where
each sub-block is coded by the CS coded scheme, and is so-called coded compressed sensing (CCS). Some redundant parity-check bits are added in the encoding phase for stitching. In the decoding process, the inner decoder is designed for recovering the collection of sub-blocks transmitted within a slot, which is amenable to CS recovery. The outer tree decoder recovers the original information sequences from all the users by stitching together valid sequences of elements drawn from the various CS lists. Authors in \cite{Enhance2020} further showed that the inner and outer decoders can be executed concurrently to achieve both performance improvement and complexity reduction. Then, a novel framework that facilitates dynamic interactions between the inner and outer decoding was presented in \cite{2020Vamsi} based on the message passing theory.

The authors in \cite{2019ale} extended the low-complexity CCS framework into a massive multiple input multiple output (MIMO)  case, since the MIMO provides extra dimensions for the received signals to provide higher detection accuracy and more user access. It uses the same outer tree code as \cite{Amalladinne2020} to stitch the information sequences in all the blocks, and a non-Bayesian inner decoder is designed based on a covariance based-maximum likelihood (CB-ML) algorithm. It shows that the affordable number of users can be larger as the number of the receiver antennas increase, hinting the massive MIMO can provide extra spectrum efficiency in URA. However, there are still some redundant parity bits inserted in the encoding process, which highly
reduces the overall spectral efficiency. For this issue, authors in \cite{Shyianov2021} proposed a framework to eliminate the need of redundant parity bits and better utilize the massive MIMO channel. The core idea of \cite{Shyianov2021} is that the block-wise information sequences can be stitched together by clustering the estimated channel vectors produced by the inner approximate message passing (AMP) decoder. This work leads to an important revelation in massive MIMO URA that no parity bits are required when the BS is equipped with a large number of antennas, since the channel vectors can be regarded as virtual signatures of user devices to connect the pieces of information sequences. However, the reducing length of each sub-block for decreasing the computational overhead of the BS, will increase the probability of user codeword collisions, i.e., more than one device chooses the same codeword. The possible codewords collisions in each block will further result in a superimposition of the estimated channel vectors, leading to the failure of the stitching process. Besides, all these CS-based divide-and-conquer algorithms map a piece of data with each sub-block to a long transmitted codeword, which is directly proportional to the number of the active users. This leads to a very low efficiency as the number of the sub-blocks are large. In addition to the divide-and-conquer strategies, there are some alternative low-complexity schemes in the context of MIMO URA \cite{Tensor, tanner}. The authors in \cite{Tensor} offered an efficient solution based on a tensor construction. In \cite{tanner}, the channel coherence interval was divided into a number of sub-slots and each active transmitter selected certain sub-slots to repeatedly transmit its codeword according to a sparse Tanner graph. At the receiver, a novel iterative algorithm was designed for decoding.

Recently, there is another line of works in \cite{PilotFengler, PilotAhmadi, FASURA} that has investigated the use of massive MIMO in the context of URA. These works have shown great efficiencies in the context of MIMO URA  with the quasi-static fading channels. Concretely, they divided the data into two phases, which are also called as the pilot-based approaches. The method in \cite{PilotFengler} is the original version of the pilot-based methods, where channel estimation is obtained in the pilot phase, and the decoding algorithm is performed in the decoding phase by considering the channel coefficients are known. The polar code associated with the list decoder has been adopted in the decoding phase. The method in \cite{PilotAhmadi} extended the work \cite{PilotFengler}, where the multiple stages of orthogonal pilots and an additional successive interference cancellation (SIC) operation are adopted to cope with the user collisions in the pilot phase and enhance the decoding performance.
The method in \cite{FASURA} combined the pilot-based method with an additional operation called noisy pilot channel estimation (NOPICE) to re-estimate the channel based on the observed signals in both the pilot and the data phases, and further re-update the estimations of the coded symbols.

Along this line, we design a novel two-phase CS-based massive MIMO URA scheme, where the device's transmitted data are divided into two phases with CS coded. No additional parity bits and stitching process are required for our scheme. The decoding of sub-blocks are established as a series of CS problems. More importantly, from the perspective of the CS, we provide a theoretical framework to analyze the induced problem based on the replica method. The analytical results show that the performance metrics of our URA scheme can be evaluated by a single-valued free entropy function. Moreover, when the massive MIMO channel coefficients of user devices are provided, each piece of data from user devices can be mapped to a short transmitted codeword. The overall length of the codewords approximately linearly increases, as the increase of the number of the information bits. For practical applications, we further design a novel AMP-based algorithm with an additional noise tuning operation.

Our contributions can be summarized as follows.
\begin{itemize}
\item
This paper proposes a two-phase URA scheme with CS coding. In the first phase, we collect a long sub-block of information bits. All the user CSI and the associated information bits are recovered within this phase by a standard CS approach. In the second phase, the residual message sequences of all the users are partitioned into sub-blocks with a very short length, which exhibits a higher spectral efficiency and a lower computational complexity than the existing CS-based transmission schemes in URA. By introducing a common codebook with mutually orthogonal codewords and taking advantage of the knowledge of estimated CSI in the first phase, the sub-block recovery problem in the second phase is cast as a support recovery problem.

\item
We then establish a performance analysis framework and performance metrics for our proposed URA scheme under various system parameters, including the number of antennas, the number of user devices, and the length of the sub-blocks, etc. This is achieved by adopting an analytical tool in the statistical physics literature, called replica method \cite{2012Probabilistic, RepMMVzhu2018, RepMMVG2018}. Concretely, we derive the replicated free entropy function over the designed system. Based on that, we establish the connection between the Bayes-optimal error probability and the system parameters. The phase transition diagram of the system can also be depicted based on the free entropy function.

\item
A decoder tailored to the induced recovery problem in the second phase is proposed based on the core idea of hybrid AMP, which achieves the performance indicated by the proposed performance analysis framework. To cope with the non-ideal CSI required in the first phase, we modify the proposed-AMP algorithm by a noise variance tuning in each iterative step. Simulations verify that our scheme outperforms the most recent counterparts.
\end{itemize}


\section{System Model}
Consider a single-cell cellular network consisting of $K_{\rm tot}$ single-antenna user devices served by a BS equipped with $M$ antennas, located at the center of a cell. We suppose there is only a fraction of number of user devices are active simultaneously, and the total number of active user devices is $K$. Let $\mathcal S$ denotes the set of user devices, and $|\mathcal S| = K$. All the user devices are assumed uncoordinated, where each single user device expects to transmit a signal, loaded with $B$ bits of information to the BS with the signal from the $i$th user device is denoted as $\mathbf b_i\in\mathbb C^{T_{\rm tot}\times 1}$, $i\in \mathcal S$. We define the massive MIMO channel from the $i$th user device to the BS as $\mathbf h_i\in\mathbb C^{M\times 1}, i\in \mathcal S$. Each user channel $\mathbf h_i$ is supposed to follow an independent Rayleigh fading with path-loss and shadowing component denoted as $\beta_i$, i.e., $\mathbf h_i\sim\mathcal N_{\mathbb C}(0,\beta_i\mathbf I_M)$. We note that this paper adopts a quasi-static channel model, where all the user channels are considered invariant within one coherence block. Since it is typical in the mMTC applications that the length of user packets and the duration for a transmission are very short, we assume all the user devices finish one transmission within one coherence block through invariant channels. The BS adopts a slotted structure for multiple access on the uplink through synchronization\footnote{We notice that the synchronization assumption is a general consideration in multiple access \cite{Amalladinne2020, Enhance2020, 2020Vamsi, Fengler2019, Shyianov2021}. For the case of the asynchronous massive connectivity, please refer to the recent work \cite{meixia2021}.}, yielding the following signal model at the BS
\begin{equation}
\tilde{\mathbf Y} = \sum\limits_{i\in \mathcal S} \mathbf b_i\mathbf h_i^T+\tilde{\boldsymbol\Xi},\label{rec_mod}
\end{equation}
where $\tilde{\boldsymbol\Xi}\in\mathbb C^{T_{\rm tot}\times M}$ is AWGN with each element distributed as $\mathcal N_{\mathbb C}(0,\sigma^2)$. Based on the received signal $\tilde{\mathbf Y}$, the decoder at the BS aims to provide an unordered list of transmitted information bit sequences. The core mission of the URA is to determine the decoder that fulfills the target fraction of error information bit sequences
with manageable complexity and the highest possible spectral efficiency.

\section{Proposed URA Scheme}
The innovation of an URA scheme lies in the map between the information bit sequences and the transmitted signals. Before proceeding, we briefly review the URA schemes.
\subsection{A Brief Review of URA Schemes}
Typically, we map each possible information bit sequence into a unique codeword, resulting in a common codebook $\tilde{\mathbf C} \in\mathbb C^{T_{\rm tot}\times 2^B}$. This yields that the system model (\ref{rec_mod}) can be expressed in the following matrix form
\begin{equation}
\tilde{\mathbf Y} = \tilde{\mathbf C}\tilde{\mathbf X}\text{diag}^{\frac{1}{2}}(\tilde{\boldsymbol\rho})\mathbf H^T+\tilde{\boldsymbol\Xi},\label{sysm1}
\end{equation}
where $\mathbf H\triangleq [\mathbf h_1,\dots, \mathbf h_K]$ is the collection of user channel vectors and $\tilde{\boldsymbol\rho}\triangleq\{\tilde{\rho}_k\}_{k = 1}^K$ is the collection of transmit powers of all the user devices. We note $\tilde{\mathbf X}$ is a $2^B\times K$ matrix with only one non-zero entry in each column, i.e., $\tilde{\mathbf x}_k\in\{0,1\}^{2^B\times 1}$, which is an indicator matrix that indicates the transmitted codewords of all the user devices. Above formulation implies a close connection between the URA and the CS problem.  However, the computational overhead of such the CS problem increases exponentially with the length of information bit sequences. Hence, the corresponding CS algorithm is computationally prohibitive even with a moderate value of $B$, preventing its application in URA.

%




The main motivation for the URA schemes \cite{Amalladinne2020, Enhance2020, 2020Vamsi, Shyianov2021} is to manage the complexity. It divides the data stream generated by each active device into $\bar J$ sub-blocks with the same length $B/\bar J$. The signals in sub-blocks are transmitted sequentially, with every sub-block becoming a new instance of URA with the received signal model consistent with (\ref{sysm1}), albeit one with a much reduced number of information bits of length $B/\bar J$. The dimension of codebook then becomes $n_c\times 2^{B/\bar J}$, collecting the codewords of length $n_c$.
After completing all these consecutive instances, the original message is recovered by splicing together the fragments of bit sequences associated with particular user devices. Such a URA scheme reduces the computing complexity at BS, since the required complexity is proportional to $J2^{B/\bar J}$ instead of $2^B$. However, the data fragment in each sub-block is mapped into a long codeword in a similar scale of the number $K$ of the total user devices to guarantee the accurate CS recovery. Therefore, the increasing number of each sub-blocks also brings a lower spectral efficiency. To handle these issues, we propose a novel two-phase URA scheme, which better uses the massive MIMO channel to enhance both the spectral and computational efficiencies compared with the existing divide-and-conquer massive MIMO URA approaches.

\subsection{Proposed Collision-supporting URA Scheme}
 For our pattern, the $B$ bits of information sequences are partitioned into two parts with a total of $J$ CS-coded sub-blocks. The former part with $L_0$ information bits are coded with the codeword of length $n$. Based on the codebook, the BS jointly recovers the information bits and the channel coefficients of user devices. In the second phase, $(J-1)$ very short sub-blocks of length $T$ loaded with a very small number of information bits of length $L$ are consecutively transmitted. The information bits in the second phase are decoded based on the CSI acquired in the first phase. Accordingly, the proposed two-phase URA scheme is shown in Fig. \ref{sys}. We elaborate on these two encoding phases below.

\begin{figure}[ht]
\centering
\includegraphics[width=3.5in]{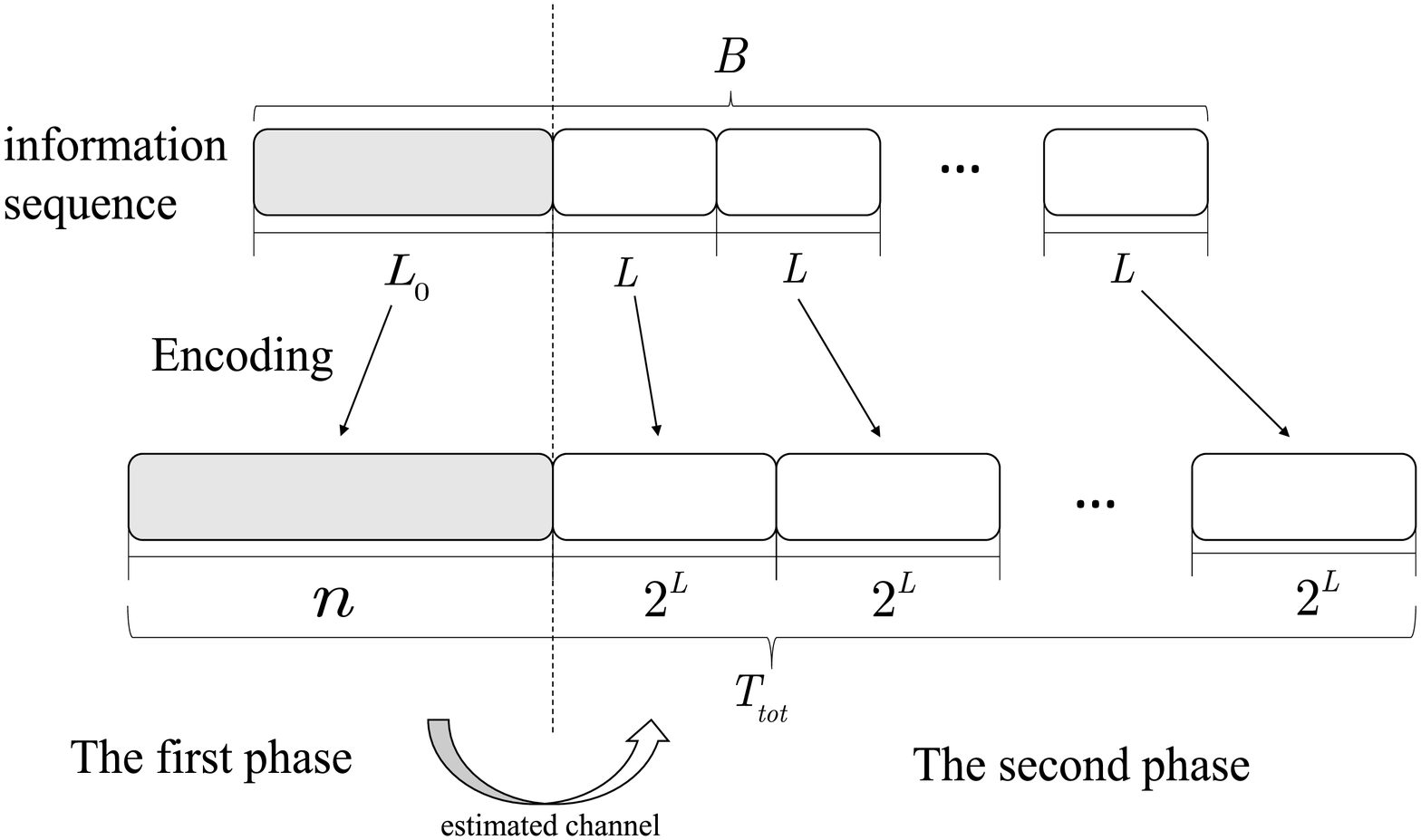}
\caption{The proposed two-phase URA divide-and-conquer strategy.}
\label{sys}
\end{figure}

The signal model for the first phase in the receiver is consistent with (\ref{rec_mod}) by replacing $T_{\rm tot}$ and $B$ with $n$ and $L_0$, respectively. We use a relatively long length of $L_0$ to avoid the codeword collisions to happen in the first phase.\footnote{In practice, the possible collisions with a very low probability can be identified by the protocol in [22], declare them as errors and retransmit them in the next frame after some feedback procedure.} The joint original message recovery and channel estimation can be formulated as a CS problem. After completing the CS instance in the first phase, the recovery signals form an unordered list of channel vectors of the $K$ active user devices. We note that this paper will not go to details of this part of work, since it is similar with the cases considered in \cite{Shyianov2021, 2022Li}, and the CS instance can be completed by directly using the algorithms in \cite{LiuMassive12018, RepMMVG2018}.

We now concentrate on the second phase of our proposed URA scheme. In each $l$th sub-block, we use a common codebook $\mathbf C\in\mathbb C^{T\times 2^L}$ and define a $2^L\times K$ index matrix $\underline{\mathbf X}_l$, where the $k$th row $\underline{\mathbf x}_{l,k}\in\{0,1\}^{2^L\times 1}$ has only one non-zero element. The received signal in each sub-block of the second phase can be formulated as
\begin{equation}
\underline{\mathbf Y}_l = \mathbf C\underline{\mathbf X}_l\text{diag}^{\frac{1}{2}}(\tilde{\boldsymbol\rho})\mathbf H^T+\underline{\boldsymbol\Xi}_l,\label{system}
\end{equation}
where $\underline{\boldsymbol\Xi}_l$ is AWGN with each element distributed as $\mathcal N_{\mathbb C}(0,\sigma^2)$. Different from the signal model (\ref{sysm1}), we consider the channel matrix $\mathbf H$ is known in the second phase after executing the joint recovery algorithm in the first phase. As we shall explain in the Sec. \ref{ipcsi}, the channel estimation error can be absorbed into the noise matrix.

Another critical difference between the first and the second phase is the length of information bits in each sub-block. We adopt a very short length of information bits for each sub-block in the second phase, so that it is general the case that $K\gg 2^L$. This leads to that the row-sparsity of $\underline{\mathbf X}_l$ disappears, and we can naturally adopt a codebook with mutually orthogonal codewords, i.e., $T = 2^L$, and $\mathbf C^H\mathbf C = \mathbf I_{2^L}$. Then, the problem of estimating the matrix $\underline{\mathbf X}_l$ with given $\mathbf C$ and $\mathbf H$ in each sub-block can be formulated as a new CS instance. Multiplying from the left by $\mathbf C^H$, neglecting the sub-block index $l$ and defining $\mathbf Y^T\triangleq\mathbf C^H\underline{\mathbf Y}$ and $\boldsymbol\Xi^T\triangleq\mathbf C^H\underline{\boldsymbol\Xi}$, we obtain from (\ref{system}) that
\begin{equation}
\mathbf Y = \mathbf S\mathbf X+\boldsymbol\Xi,\label{equi_mod}
\end{equation}
where $\mathbf Y\in\mathbb C^{M\times 2^L}$, $\mathbf S\in\mathbb C^{M\times K}$ is the known normalized channel matrix with all the elements $s_{m,k}$ i.i.d. distributed as $\mathcal N_{\mathbb C}(0,\frac{1}{M})$, $\mathbf X$ is a $K\times 2^L$ matrix with only one non-zero entry in each row, i.e., $\mathbf x_k\in\{0,\rho_k\}^{1\times 2^L}$, $\rho_k = \sqrt{M\tilde{\rho_k}\beta_k}$, and $\boldsymbol\Xi$ is AWGN with each variance $\sigma^2$.

Compared with the existing CS-based URA schemes \cite{Amalladinne2020, Enhance2020, 2020Vamsi, Shyianov2021}, our proposed two-phase URA scheme has the following advantages. First, such a consideration is greatly beneficial for the computational efficiency at the BS, since the complexity in each sub-block increases exponentially with the length of information bits. Second, the overall number of channel uses of the proposed URA scheme is $n+(B-L_0)/L\times 2^L$, which also becomes smaller with the decreasing $L$. When the value of $L$ is small, the overall length of the codewords in the second phase is in the same scale as $B$. This is a very significant advantage compared with the schemes in \cite{Amalladinne2020, Enhance2020, 2020Vamsi, Shyianov2021}. As aforementioned, they map each sub-block of data to a long transmitted codeword, the length of which is proportional to $K$, leading to the overall length of the codewords typically far larger than $B$. Besides, the channel matrix $\mathbf S$ can be regarded as signatures associated with particular user devices. This results in that no additional parity bits and stitching process are required for our pattern. In addition, we can see in the subsequent sections that the codeword collisions induced by the short length of the sub-blocks in the second can be efficiently coped with by leveraging the CSI knowledge acquired in the first phase.

The above discussions illustrate the advantages of our URA scheme in terms of complexity and spectral efficiency. Based on that, one critical mission is to evaluate the decoding performance of the pattern. According to (\ref{equi_mod}), our task in the second phase turns into estimating the matrix $\mathbf X$ based on the normalized channel matrix $\mathbf S$ and transformed received signal $\mathbf Y$. We determine the position of the non-zero element in the $k$th row $\mathbf x_k$ of $\mathbf X$ to determine the original message associated with the channel signature $\mathbf h_k$. We note that the $\mathbf h_k$ indicates the channel vector that corresponds to the $k$th column in $\mathbf S$.
Intuitively, considering the information bit sequences of all the user devices are produced uniformly, there are about $K/2^L$ user devices that transmit the same information bit sequence of length $L$ in each sub-block. Accordingly, the new CS instance induced by (\ref{equi_mod}) is a multiple measurement vectors (MMV)-CS problem with about $K/2^L$-sparsity for each column of matrix $\mathbf X$. From the perspective of the CS, the support recovery performance is impacted by the number of the BS antennas $M$, and the level of sparsity $K/2^L$. Hence, it can be inferred that the decreasing $L$ is beneficial for the computational overhead reduction as well as the spectral efficiency improvement, but could increase the probability of the codeword collisions and so that reduces the decoding performance with given number of the BS antennas $M$ and the active user devices $K$. Accordingly, in order to facilitate the practical applications of our URA scheme, we need to further evaluate how short can the length $L$ be adopted to guarantee a certain error probability threshold in a given communication system with specific system parameters. For this purpose, we propose a theoretical framework in the next section, which analytically reveals the connection between the decoding performance of our proposed URA scheme and above system parameters.




\section{Performance Analysis for Proposed URA Scheme}
In this section, we propose a theoretical framework to analyze the performance metrics of our proposed URA scheme in the Bayes-optimal setting. It provides the Bayes-optimal error probability performance of our scheme with respect to the system parameters. Meanwhile, this performance analysis framework also indicates the direction of our decoder design, which will be investigated in the following sections. As aforementioned, the channel estimation error in the first phase can be absorbed into the noise matrix $\boldsymbol\Xi$, and we thus focus on the signal model (\ref{equi_mod}) here.


To measure the performance of our URA scheme, we investigate two performance metrics, i.e., the per-user error probability, and the  mean square error (MSE), which are two general metrics in the URA literature. The per-user error probability captures the fraction of incorrect recovered information bit sequences at the BS, which can be formulated based on (\ref{equi_mod}) as
\begin{equation}
P_e = \frac{1}{K}\sum\limits_{k = 1}^K \mathbb I(\exists l = 2,\dots, J, \mathbf x_{k}\neq \bar{\mathbf x}_{k}),\label{def_Pe}
\end{equation}
where $\mathbb I(\cdot)$ is the indicator function, $l$ is the index of sub-block, $\hat{\mathbf X}\triangleq [\hat{\mathbf x}_1^T,\dots, \hat{\mathbf x}_K^T]^T$ is the estimate of $\mathbf X$ and $\bar{\mathbf X}\triangleq [\bar{\mathbf x}_1^T,\dots, \bar{\mathbf x}_K^T]^T$ with each column $\bar{\mathbf x}_k$ obtained after hard thresholding of the corresponding $\hat{\mathbf x}_k$. On the other hand, the MSE can be formulated as
\begin{align}
&\text{MSE} = \frac{1}{K}\sum\limits_{k = 1}^K||\mathbf x_k-\hat{\mathbf x}_k||^2.\label{met}
\end{align}


In the following, we propose an analytical framework to evaluate the performance metrics defined in (\ref{def_Pe}) and (\ref{met}) based on the system parameters. Our analysis in this paper is under the Bayesian framework. Accordingly, we associate a posterior probability $p(\hat{\mathbf X}|\mathbf Y)$ to the signal estimation given the received signal, which can be given based on the Bayes formula by
$
p(\hat{\mathbf X}|\mathbf Y) =
\frac{p(\hat{\mathbf X})p(\mathbf Y|\hat{\mathbf X})}{\int{\rm d}\hat{\mathbf X}p(\hat{\mathbf X})p(\mathbf Y|\hat{\mathbf X})},
$
where the likelihood based on model (\ref{equi_mod}) is given by
\begin{equation}
p(\mathbf Y|\hat{\mathbf X}) = \frac{1}{(\pi\sigma^2)^M}\exp\left(-\frac{(\mathbf Y-\mathbf S\hat{\mathbf X})(\mathbf Y-\mathbf S\hat{\mathbf X})^H}{\sigma^2}\right),\label{pos}
\end{equation}
and the prior of $\hat{\mathbf X}$ is given by $p(\hat{\mathbf X}) = \prod\nolimits_{k = 1}^Kp(\hat{\mathbf x}_k)$, where each factor is formulated as
\begin{equation}
p(\hat{\mathbf x}_k) = \frac{1}{2^L}\sum\nolimits_{i = 1}^{2^L}\delta(\hat{x}_{k,i}-\rho_k)\prod_{m\neq i}\delta(\hat{x}_{k,m}),\label{con_pri}
\end{equation}
where $\delta(\cdot)$ is a Dirac delta function and  this prior distribution enforces each vector to have only one non-zero value, indicating the codewords transmitted by user devices. The probability $1/2^L$ indicates that information bit sequences of the user devices in each sub-block are all uniformly distributed.

We next adopt the replica method \cite{RepTan2002, 1985Entropy} to evaluate the free entropy function of this Bayesian probability model, which reflects the macro performances of our system (\ref{equi_mod}), revealing the connection between the system parameters and the performance metrics. Typically, our statistical approach for analysis is performed in a certain asymptotic regime with $K\to\infty$, $M\to\infty$ and a fixed ratio $M/K = \alpha$. For the ease of concise expression, similarly with the scenario considered in \cite{AMPdecoder}, we suppose that the large-scale fading of each user device is compensated by the user transmit power, which normalizes the noise variance, i.e., we consider $\rho_k = 1$ and the power factor is absorbed into the parameters $\sigma^2$. We remark that although we have restricted ourselves to a particular power assumption, the analysis for a generic power allocation is straightforward by further considering an empirical distribution over the power sequence $\{\rho_k\}_{k = 1}^K$.

\subsection{Evaluating MSE based on the Free Entropy Function}
In order to derive the free entropy function, we begin by writing the partition function over the probability $p(\mathbf Y|\hat{\mathbf X})$ defined in (\ref{pos}) as
\begin{equation}
\mathcal Z(\mathbf Y, \mathbf S) = \sum\limits_{\hat{\mathbf X}}p(\hat{\mathbf X})\exp\left(-\frac{(\mathbf Y-\mathbf S\hat{\mathbf X})(\mathbf Y-\mathbf S\hat{\mathbf X})^H}{\sigma^2}\right),\label{part}
\end{equation}
and the free energy $\mathcal F_K(\mathbf Y, \mathbf S)$ is defined by
$
\mathcal F_K(\mathbf Y, \mathbf S) \triangleq K^{-1}\log\mathcal Z(\mathbf Y, \mathbf S)
$. We further define the limit $\mathcal F\triangleq \lim\limits_{K\to\infty} \mathcal F_K(\mathbf Y, \mathbf S)$. It is supposed that the limit exists and is equivalent to its average for almost all realizations of the user channels and the noise, i.e., we sequently have
\begin{equation}
\mathcal F = \lim\limits_{K\to\infty}K^{-1}\mathbb E_{\mathbf S,\mathbf Y}(\log\mathcal Z(\mathbf Y, \mathbf S)).\label{S-A}
\end{equation}
In order to evaluate the free entropy (\ref{S-A}), we make use of the treatment in \cite{1985Entropy}. When $n$ is small enough, we sequently have
\begin{align}
\mathcal F &= \lim\limits_{K\to\infty}\lim\limits_{n\to 0}K^{-1}\frac{\partial\log\mathbb E_{\mathbf S,\mathbf Y}(\mathcal Z(\mathbf Y, \mathbf S))^n}{\partial n}\nonumber\\
&= \lim\limits_{K\to\infty}\lim\limits_{n\to 0}\frac{\mathbb E_{\mathbf S,\mathbf Y}(\mathcal Z(\mathbf Y, \mathbf S))^n-1}{Kn}.\label{rep_tri}
\end{align}
We note that the term $\mathcal Z(\mathbf Y, \mathbf S)$ in (\ref{rep_tri}) is the partition function of the posterior distribution (\ref{pos}), and $(\mathcal Z(\mathbf Y, \mathbf S))^n$ is the so-called replicated partition function as it can be regarded as the partition function associated with $n$ replicas that are all independently drawn from (\ref{pos}). The general replica assumptions are provided in \cite{1985Entropy}, which are also used in this paper, given by
\begin{itemize}
\item
According to the replica trick \cite{1985Entropy}, the quantity $\mathbb E_{\mathbf S,\mathbf Y}(\mathcal Z(\mathbf Y, \mathbf S))^n$ in (\ref{rep_tri}) is calculated as if $n$ were an integer, and we take the fact that $n$ is a real number into consideration after obtaining a tractable enough expression.
\item
The order of the two limits $K\to\infty$ and $n\to 0$ in (\ref{rep_tri}) can be interchanged without affecting the final result of the free entropy calculation.
\item
All the replicas are statistically equivalent, and thus the overlaps (some parameters that will be defined in the following) are independent of the replica indices, which is also called replica-symmetry (R-S) assumption.
\end{itemize}
Based on that,  we can arrive at the following theorem. We note that the derivations basically follow the style of the replica analysis in \cite{1985Entropy}. We heuristically extend the R-S assumption, from the single measurement vector (SMV) case to the MMV case. The R-S assumption in the MMV case is not rigorously proved, but its validity has been demonstrated in \cite{RepMMVzhu2018, RepMMVG2018}.

\setcounter{TempEqCnt}{\value{equation}}
\setcounter{equation}{11}
\begin{figure*}[tbp]
\begin{align}
 &\bar{\Phi}(\mathbf E) = {\rm Tr}\left(-\left(\boldsymbol\Delta+\frac{1}{\alpha}\mathbf E\right)^{-1}\left(\alpha\boldsymbol\Delta+\mathbf C\right)\right)-\alpha\log\left|\boldsymbol\Delta+\frac{1}{\alpha}\mathbf E\right|+\sum\limits_{\mathbf x}p(\mathbf x)\int D\mathbf z\log\left(\sum\limits_{\hat{\mathbf x}}p(\hat{\mathbf x})\right.\nonumber\\
 &\times\left.\exp\left({\rm Tr}\left(2\left((\boldsymbol\Delta+\frac{1}{\alpha}\mathbf E)^{-1}\mathbf x^T+(\boldsymbol\Delta+\frac{1}{\alpha}\mathbf E)^{-\frac{1}{2}} \mathfrak R(\mathbf z^T)\right)\hat{\mathbf x}-(\boldsymbol\Delta+\frac{1}{\alpha}\mathbf E)^{-1}\hat{\mathbf x}^T\hat{\mathbf x}\right)\right)\right),\label{fre_ent_E}
\end{align}
\setcounter{equation}{14}
\begin{align}
\Phi(d) = &-\frac{d+2^L\alpha\sigma^2+1}{\sigma^2+\alpha^{-1}d}-(2^L-1)\alpha\log(\sigma^2+\alpha^{-1}d)+\int D\mathbf z\log\Bigg(\sum\limits_{i = 2}^{2^L}\nonumber\\
&\exp\bigg(-\frac{1}{\sigma^2+\alpha^{-1}d}+\frac{2z_i}{(\sigma^2+\alpha^{-1}d)^{\frac{1}{2}}}\bigg)
+\exp\bigg(\frac{1}{\sigma^2+\alpha^{-1}d}+\frac{2z_1}{(\sigma^2+\alpha^{-1}d)^{\frac{1}{2}}}\bigg)\Bigg).\label{phi2}
\end{align}
\hrulefill
\end{figure*}
\setcounter{equation}{\value{TempEqCnt}}

\begin{theorem}
Consider a Bayes-optimal setting, where the experimental distribution $p(\boldsymbol X)$ over all the rows in matrix $\mathbf X$ matches the considered distribution (\ref{con_pri}) with the second-order moment denoted as $\mathbf C\triangleq\mathbb E_{\boldsymbol X}(\boldsymbol X^H\boldsymbol X)$. Let all the replica assumptions hold, by adopting the saddle point method \cite{2012Probabilistic}, the free entropy in (\ref{rep_tri}) can be calculated by a function with respect to a $2^L\times 2^L$ matrix $\mathbf E$, given by (\ref{fre_ent_E}) at the top of this page, where $\boldsymbol\Delta\triangleq \sigma^2\mathbf I$, the term $D\mathbf z$ is a Gaussian integration measure, and the notation $\mathfrak R(\cdot)$ returns the real part of the input. The function in (\ref{fre_ent_E}) is then called free entropy function.
\end{theorem}

\begin{proof}
Please see the Appendix \ref{ap1}.
\end{proof}
\setcounter{equation}{12}
We define a matrix
\begin{equation}
\mathbf E^{\star}\triangleq \frac{1}{K}\sum\limits_{k = 1}^K(\mathbf x_k-\hat{\mathbf x}_k)^H(\mathbf x_k-\hat{\mathbf x}_k),\label{E}
\end{equation}
where each $\hat{\mathbf x}_k$ is the sample of the posterior distribution $p(\hat{\mathbf X}|\mathbf Y)$ in (\ref{pos}). Hence, in the large $K$ limit, the sum of the diagonal elements of $\mathbf E^{\star}$ corresponds to the minimum MSE (MMSE) of our system. According to the basic theory in \cite{1985Entropy}, the matrix $\mathbf E^{\star}$ is exactly the global optimal point of (\ref{fre_ent_E}). We can thus evaluate the MMSE performance with given system parameters by finding the global optimal point of function (\ref{fre_ent_E}). However, it is typically not easy. To simplify the expression (\ref{fre_ent_E}), we combine the heuristics of the R-S assumption, and notice that the matrixes $\mathbf C$ and $\mathbf E$ appeal some symmetrical structures. Heuristically, we assume the following parametric forms with parameters $a$, $b$ and $c$.
\begin{align}
\mathbf C = c\mathbf I_{2^L},\quad \mathbf E = (a-b)\mathbf I_{2^L}+b\boldsymbol\amalg_{2^L},\label{symm}
\end{align}
where $\boldsymbol\amalg_{2^L}$ stands for the $2^L\times 2^L$ matrix with elements all equal to one. Such structures come from the symmetry of the distribution (\ref{con_pri}), where the circumstance of a particular element is independent of its index. Accordingly, we have the following proposition that provides a simpler expression of (\ref{fre_ent_E}).
\begin{proposition}
Considering (\ref{symm}) and using $d\triangleq a-b$, the free entropy function (\ref{fre_ent_E}) can be reformulated with respect to the single scalar argument $d$, which is given by (\ref{phi2}) at the top of this page. The matrix $\mathbf E^{\star}$ that reflects the MMSE over system defined in (\ref{equi_mod}) can be evaluated by $\mathbf E^{\star} = d^{\star}\mathbf I_{2^L}-1/2^Ld^{\star}\boldsymbol\amalg_{2^L}$, where $d^{\star}$ is the global optimal point of (\ref{phi2}).
\end{proposition}
\begin{proof}
Please see the Appendix \ref{ap2}.
\end{proof}
Hence, we can evaluate the MMSE performance with given system parameters by finding the global optimal point of a single-valued function (\ref{phi2}). After that, we then investigate the connection between the per-user probability $P_e$ and the system parameters in the following section.
\setcounter{equation}{15}

\subsection{Evaluating Per-user Error Probability}
In this section, we further reveal the connection between the per-user probability $P_e$ and the system parameters by leveraging the obtained free entropy function. Before proceeding, we provide some definitions for the ease of expressions. First, we define the following signal model
\begin{equation}
\mathbf r = \mathbf x+\mathbf z\boldsymbol\Sigma^{\frac{1}{2}},\label{rmod}
\end{equation}
where $\mathbf x$ is a random variable with prior distribution (\ref{con_pri}), $\mathbf z$ is a standard normal random variable, and $\boldsymbol\Sigma\triangleq \boldsymbol\Delta+\alpha^{-1}\mathbf E^{\star}$. Then, we define a function $\boldsymbol\eta(\mathbf r, \boldsymbol\Sigma)$, which returns
\begin{equation}
\boldsymbol\eta(\mathbf r, \boldsymbol\Sigma) =\frac{\sum\nolimits_{\hat{\mathbf x}}p(\hat{\mathbf x})\mathcal N_{\mathcal C}(\hat{\mathbf x};\mathbf r, \boldsymbol\Sigma)\hat{\mathbf x}}{\sum\nolimits_{\hat{\mathbf x}}p(\hat{\mathbf x})\mathcal N_{\mathcal C}(\hat{\mathbf x};\mathbf r, \boldsymbol\Sigma)},\label{eta}
\end{equation}
with $p(\hat{\mathbf x})$ consistent with (\ref{con_pri}). We denote $\eta_{i|j}(\mathbf r, \boldsymbol\Sigma)$ as the $i$th output of $\boldsymbol\eta(\mathbf r, \boldsymbol\Sigma)$ when the index of the true position of the non-zero component is $j$. As a consequence, we have the following proposition.
\begin{proposition}
With the definition (\ref{rmod}), the Bayes-optimal error probability $P_e$ with maximum a posteriori (MAP)-based hard threshold can be formulated by
\begin{align}
P_e = &1-\left(1-\int D\mathbf z\mathbb I\left(\exists i\in\{2,\dots, 2^L\}:\right.\right.\nonumber\\
&\quad\quad\quad\quad\quad\quad \left.\left.\eta_{i|1}(\mathbf r, \boldsymbol\Sigma)>\eta_{1|1}(\mathbf r, \boldsymbol\Sigma)\right)\right)^{J-1},\label{exp_pe}
\end{align}
where $(J-1)$ is the number of sub-blocks in the second phase of our proposed URA scheme.
\end{proposition}
\begin{proof}
Please see the Appendix \ref{ap3}.
\end{proof}

The equation (\ref{exp_pe}) establishes the relationship between the per-user error probability $P_e$ and the free entropy function (\ref{phi2}). By combining the (\ref{exp_pe}) with (\ref{phi2}), we establish an analytical tool to measure the Bayes-optimal $P_e$ of our proposed URA scheme. The performance of our scheme under specifical system parameters can be then evaluated.


 Another issue is to achieve the Bayes-optimal performance revealed by (\ref{phi2}). This requires to obtain the sample of posterior distribution (\ref{pos}). Since obtaining its explicit form is always not tractable in practical applications, it typically turns to obtain an approximation to approach the Bayes-optimal performance. The AMP framework is typically considered for this purpose. It is indicated in \cite{2012Probabilistic} that the recovery performance of the AMP decoder is reflected on the largest fixed point of the free entropy function. Accordingly, our theoretical results can also be used to evaluate the performance of the AMP-based decoder by replacing the global maximum point of the free entropy function with the largest fixed point. With the changes of the system parameters, the number of fixed points of the free entropy function (\ref{phi2}) can be changed, which is called phase transition. Such a phenomenon reveals the essential reasons of the decoding performance. For example, when the largest fixed point disappears and there becomes only one fixed point which maximizes the free entropy function, the performance of AMP decoder suddenly changes into the Bayes-optimal performance. The phase transition diagram of our scheme is provided based on (\ref{phi2}), as we will seen in the Sec. \ref{num_res}.

 \section{AMP Decoder For the Proposed URA Scheme}\label{pro_amp}
 In this section, we design the concrete AMP decoder tailored to our model (\ref{equi_mod}) in the second phase of our proposed URA scheme, which aims to decode the information bit sequences based on the CSI knowledge in the first phase.

\begin{figure}[ht]
\centering
\includegraphics[width=4.2in]{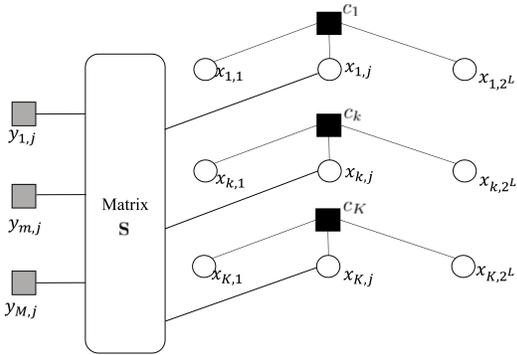}
\caption{Factor graph representation associated with (\ref{equi_mod}).}
\label{fac}
\end{figure}

\begin{algorithm*}
\centering
\caption{Proposed AMP-based algorithm}
\label{alg1}
\begin{algorithmic}[1]
\State \{Initialization\}
\State $t = 1$
\State $\forall k,j: \epsilon_{k,j}(t-1) = 1/2^L$
\State $\forall k,j: \hat x_{k,m}(t-1) = 0$
\State $\forall k,j: Q_{k,j}^x(t-1) = \epsilon_{k,j}(t-1)\rho_k$
\State $\forall l,m: \hat{s}_{m,j}(t-1) = 0$
\State Some initialization of $\sigma^2(t)$
\Repeat
\State $\forall m,j: Q_{m,j}^p(t) = \sum\nolimits_{k = 1}^K|s_{m,k}|^2Q_{k,j}^x(t-1)$
\State $\forall m,j: \hat{p}_{m,j}(t) = \sum\nolimits_{k = 1}^Ks_{m,k}\hat x_{k,j}(t-1)-Q_{m,j}^p(t)\hat{s}_{m,j}(t-1)$
\State $\forall m,j: Q_{m,j}^z(t) = Q_{m,j}^p(t)\sigma^2(t)/(Q_{m,j}^p(t)+\sigma^2(t))$
\State $\forall m,j: \hat z_{m,j}(t) = (y_{m,j}Q_{m,j}^p(t)+\hat{p}_{m,j}(t)\sigma^2(t))/(Q_{m,j}^p(t)+\sigma^2(t))$
\State $\forall m,j: Q_{m,j}^s(t) = Q_{m,j}^p(t)^{-1}\left(1-Q_{m,j}^p(t)^{-1}Q_{m,j}^z(t)\right)$
\State $\forall m,j: \hat s_{m,j}(t) = Q_{m,j}^p(t)^{-1}\left(\hat z_{m,j}(t)-\hat p_{m,j}(t)\right)$
\State $\forall k,j: Q_{k,j}^r(t) = \left(\sum\nolimits_{k = 1}^K|s_{m,j}|^2Q_{m,j}^s(t)\right)^{-1}$
\State $\forall k,j: \hat r_{k,j}(t) = \hat x_{k,j}(t-1)+Q_{k,j}^r(t)\sum\nolimits_{m = 1}^Ms^*_{m,k}\hat s_{m,j}(t)$
\State $\forall k,j: \hat x_{k,j}(t) = \mathbb E_{X_{k,j}|\boldsymbol Y_j}(x_{k,j}|\mathbf y_j; \hat r_{k,j}(t), Q_{k,j}^r(t))$
\State $\forall k,j: Q_{k,j}^x(t) = {\rm Var}_{x_{k,j}|\boldsymbol Y_j}(x_{k,j}|\mathbf y_j; \hat r_{k,j}(t), Q_{k,j}^r(t))$
\State $\forall k,j: \text{LLR}_{k,j}(t) = -\log\sum\limits_{i\neq j}\exp\left(\frac{|\hat r_{k,i}(t)|^2}{Q_{k,i}^r(t)}-\frac{|\hat r_{k,i}(t)-\rho_k|^2}{Q_{k,i}^r(t)}\right)$
\State $\forall k,j: \epsilon_{k,j}(t) =  (1+\exp(-\text{LLR}_{k,j}(t)))^{-1}$
\State $\sigma^2(t+1) = \frac{1}{2^LM}\sum\nolimits_{j = 1}^{2^L}\sum\nolimits_{m = 1}^M\left(|y_{m,j}-\hat z_{m,j}(t)|^2+Q_{m,j}^z(t)\right)$
\State $t = t+1$
\Until{$\sum\nolimits_{k = 1}^K\sum\nolimits_{m = 1}^M|\hat x_{k,m}(t+1)-\hat x_{k,m}(t)|^2>\epsilon\sum\nolimits_{k = 1}^K\sum\nolimits_{m = 1}^M|\hat x_{k,m}(t)|^2$ or $t>T_{\rm max}$}
\end{algorithmic}
\end{algorithm*}
\subsection{The AMP-based Algorithm}\label{dal}
Based on a dense graph corresponding to our system model (\ref{equi_mod}), which is specified in Fig. \ref{fac}, an AMP-based algorithm is designed to efficiently solve the problem of estimating $\mathbf X$. We denote the factor node associated with the elements of $k$th user device as $c_k(\mathbf x_k) \triangleq \frac{1}{2^L}\sum\nolimits_{j = 1}^{2^L}\delta({x}_{k,j}-\rho_k)\prod_{m\neq j}\delta({x}_{k,m})$, and we use the notation $\nu_{x_{k,j}\leftarrow c_k}$ and $\nu_{x_{k,j}\rightarrow c_k}$ to denote the message from $c_k$ to $x_{k,j}$ and from $x_{k,j}$ to $c_k$, respectively. Both $\nu_{x_{k,j}\leftarrow c_k}$ and $\nu_{x_{k,j}\rightarrow c_k}$ can be viewed as functions of $x_{k,j}$.

We note that the structure of the factor graph in Fig. \ref{fac} blocks the direct application of the AMP algorithm \cite{NOMA} designed for multi-user detection in non-orthogonal multiple access (NOMA), although our scenario in the second phase is similar with that in NOMA. We explain this from the perspective of the message passing principle. Because of the CS coding scheme, the message $\nu_{x_{k,j}\leftarrow c_k}(x_{k,j})$ depends not only on the node $c_k$, but also the message $\nu_{x_{k,j'}\rightarrow c_k}(x_{k,j'})$ from other factor nodes $x_{k,j'}, \forall j'\neq j$, resulting in that the message $\nu_{x_{k,j}\leftarrow c_k}(x_{k,j})$ will not remain unchanged during the AMP update procedure. Accordingly, we require to simultaneously update the message $\nu_{x_{k,j}\leftarrow c_k}(x_{k,j})$ in each iteration of AMP algorithm.

For this problem, we design a novel algorithm based on the core idea of the Hybrid generalized AMP (GAMP) framework \cite{HybridGAMP}. Note that although there are no rigorous proofs for the state evolution of this framework, it has shown great power in dealing with the inherent structure of a probabilistic model.
Hybrid GAMP framework partitions the edges of the graphical model into the weak subsets and strong subsets. This leads to that the messages propagating on the weak edges can be approximated by AMP-like formulas, while that on the strong edges can be calculated based on the standard loopy belief propagation (BP). The designed algorithm is carefully tailored to our system model specified by Fig. \ref{fac}, where the messages propagating from the observed nodes to $\mathbf x_k$ are considered in the weak subsets, and the messages between $c_k$ and $\mathbf x_k$ are considered in the strong subsets. We note that the $y_{m,j}$ is the $m$th element of the $j$th column of the matrix $\mathbf Y$.

For the ease of concise expressions, we work with the log-likelihood ratios (LLRs) for Bernoulli variables. Specifically, we define the LLR of each Bernoulli variable $x_{k,j}$ as
\begin{equation}
\text{LLR}_{k,j} = \log\frac{\nu_{x_{k,j}\leftarrow c_k}(x_{k,j} = \rho_k)}{\nu_{x_{k,j}\leftarrow c_k}(x_{k,j} = 0)},\label{llr}
\end{equation}
where the involved messages can be specified by the general message passing principle, given by
\begin{equation}
\nu_{x_{k,j}\leftarrow c_k}(x_{k,j}) = \sum\nolimits_{j'\neq j}c_k(\mathbf x_k)\prod\nolimits_{j'\neq j}\nu_{x_{k,j'}\rightarrow c_k}(x_{k,j'}).\label{ctox}
\end{equation}
The message $\nu_{x_{k,j'}\rightarrow c_k}$ propagating from the weak edge can be obtained by the basic GAMP update equations derived via quadratic approximation. Accordingly, we have \begin{equation}\nu_{x_{k,j'}\rightarrow c_k}(x_{k,j'}) = \mathcal{N}_{\mathcal C}(\hat r_{k,j'};x_{k,j},Q^r_{k,j'}),\label{xtoc}\end{equation} with the parameters $\hat r_{k,j'}$ and $Q^r_{k,j'}$ obtained via GAMP update equations specified in Algorithm \ref{alg1}. Plugging (\ref{xtoc}) into (\ref{ctox}) and combining the definition of node $c_k$, we obtain the following equations
\begin{align}
\nu_{x_{k,j}\leftarrow c_k}&(x_{k,j} = \rho_k) \nonumber\\
= &\frac{1}{2^L}\frac{1}{\pi^{2^L-1}\prod_{j'\neq j}Q^r_{k,j'}}\exp(-\sum\limits_{j'\neq j}\frac{|\hat r_{k,j'}|^2}{Q^r_{k,j'}}),\\
\nu_{x_{k,j}\leftarrow c_k}&(x_{k,j} = 0)\nonumber\\
 = &\frac{1}{2^L}\sum\limits_{i\neq j}\frac{1}{\pi^{2^L-1}\prod_{j'\neq j}Q^r_{k,j'}}\nonumber\\
 &\quad\times\exp(-\sum\limits_{j'\neq i,j}\frac{|\hat r_{k,j'}|^2}{Q^r_{k,j'}}-\frac{|\hat r_{k,i}-\rho_k|^2}{Q^r_{k,i}}).
\end{align}
As a result, the LLR (\ref{llr}) can be calculated as
$\text{LLR}_{k,j} = -\log\sum\nolimits_{i\neq j}\exp\left(\frac{|\hat r_{k,i}|^2}{Q^r_{k,i}}-\frac{|\hat r_{k,i}-\rho_k|^2}{Q^r_{k,i}}\right)$.
Based on the definition of LLR, we define a variable $\epsilon_{k,j}\triangleq (1+\exp(-\text{LLR}_{k,j}))^{-1}$, which can be interpreted as the prior probability that $x_{k,j} = \rho_{k}$ occurs with the knowledge of inferencing results of other nodes $x_{k,j'}, \forall j'\neq j$. Hence, the GAMP update equations of each index $j$ can be obtained by considering the current values of $\epsilon_{k,j}$, which will be further updated after executing the the GAMP update equations for all indices $j = 1\dots, 2^L$. The involved posterior marginal distribution $p_{X_{k,j}|\boldsymbol Y_j}$ can be derived as
\begin{align}
&\mathbb E_{X_{k,j}|\boldsymbol Y_j}(x_{k,j}|\mathbf y_j; \hat r_{k,j}, Q_{k,j}^r) = \rho_k\hat{P}(x_{k,j} = \rho_k),\\
&{\rm Var}_{x_{k,j}|\boldsymbol Y_j}(x_{k,j}|\mathbf y_j; \hat r_{k,j}, Q_{k,j}^r) \nonumber\\
&= \rho_k^2\left(\hat{P}(x_{k,j} = \rho_k)-\hat{P}^2(x_{k,j} = \rho_k)\right),
\end{align}
where we define a probability as
\begin{align}
\hat{P}&(x_{k,j} = \rho_k)\nonumber \\
&= \left(1+\frac{1-\epsilon_{k,j}}{\epsilon_{k,j}}\exp\left(-\frac{|\hat r_{k,j}|^2}{Q^r_{k,j}}+\frac{|\hat r_{k,j}-\rho_k|^2}{Q^r_{k,j}}\right)\right)^{-1}.\nonumber
\end{align}
After executing this algorithm, we adopt the MAP criterion to derive the estimated signal $\bar{\mathbf x}_k$ associated with channel $\mathbf h_k$. It is obtained by a hard thresholding of $\hat{\mathbf x}_k$ in each section of the second phase. Each component in $\hat{\mathbf x}_k$ corresponds to a posterior distribution that the value of the corresponding position is non-zero, and the MAP is thus implemented as follow. If $\hat{x}_{k,i}$ is the maximum among all the elements in $\hat{\mathbf x}_k$, then we let $\bar{x}_{k,i}$ to be the only one non-zero component in $\bar{\mathbf x}_k$.
\begin{figure*}
\subfigure[$L = 1$, $\bar{\sigma}^2 = 0.01$.]{
\begin{minipage}[t]{0.31\linewidth}
\centering
\includegraphics[width=2.3in]{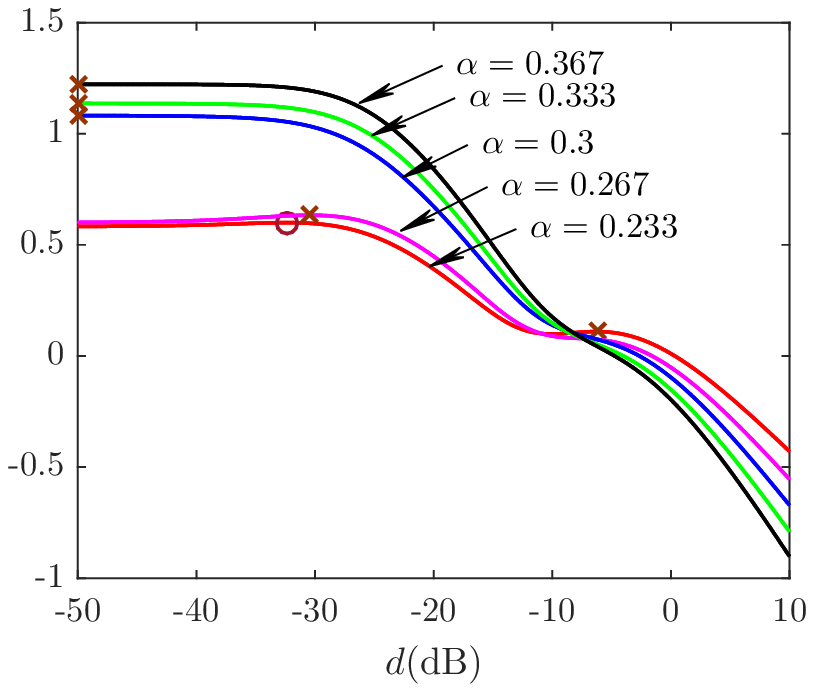}
\label{L1001}
\end{minipage}%
}
\subfigure[$L = 1$, $\bar{\sigma}^2 = 0.1$.]{
\begin{minipage}[t]{0.31\linewidth}
\centering
\includegraphics[width=2.3in]{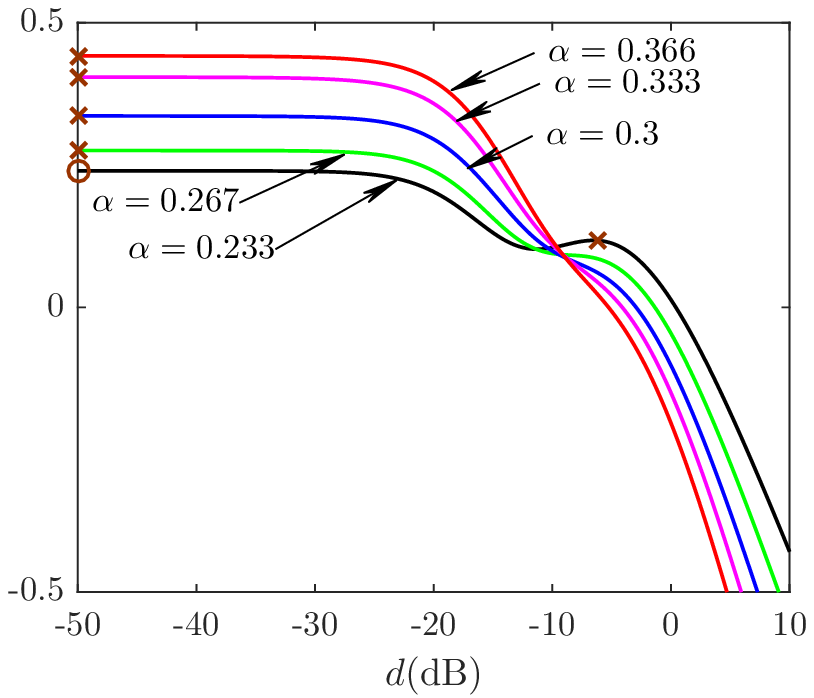}
\label{L101}
\end{minipage}
}
\subfigure[$L = 1$, $\bar{\sigma}^2 = 1$.]{
\begin{minipage}[t]{0.31\linewidth}
\centering
\includegraphics[width=2.3in]{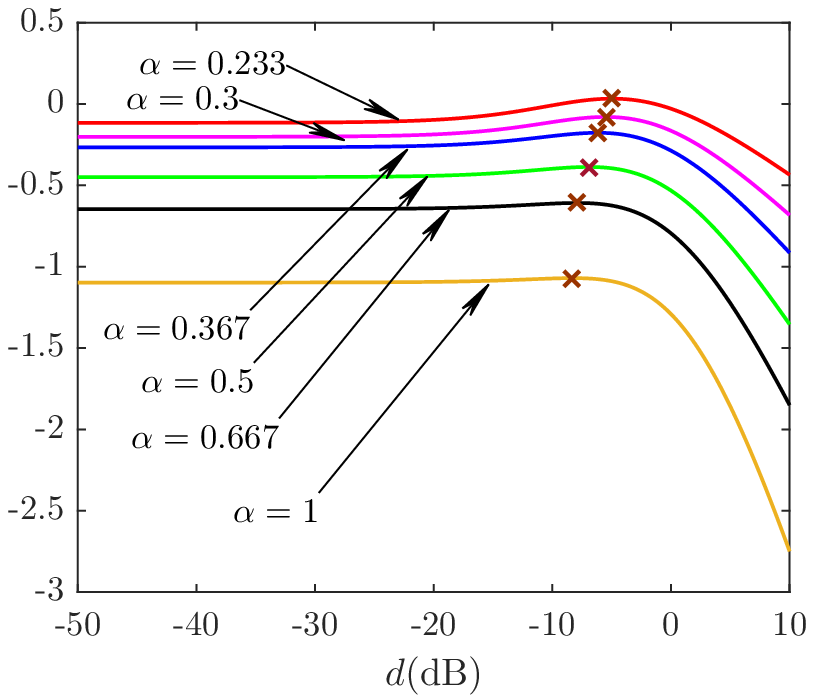}
\label{L11}
\end{minipage}
}

\subfigure[$L = 2$, $\bar{\sigma}^2 = 0.01$.]{
\begin{minipage}[t]{0.31\linewidth}
\centering
\includegraphics[width=2.3in]{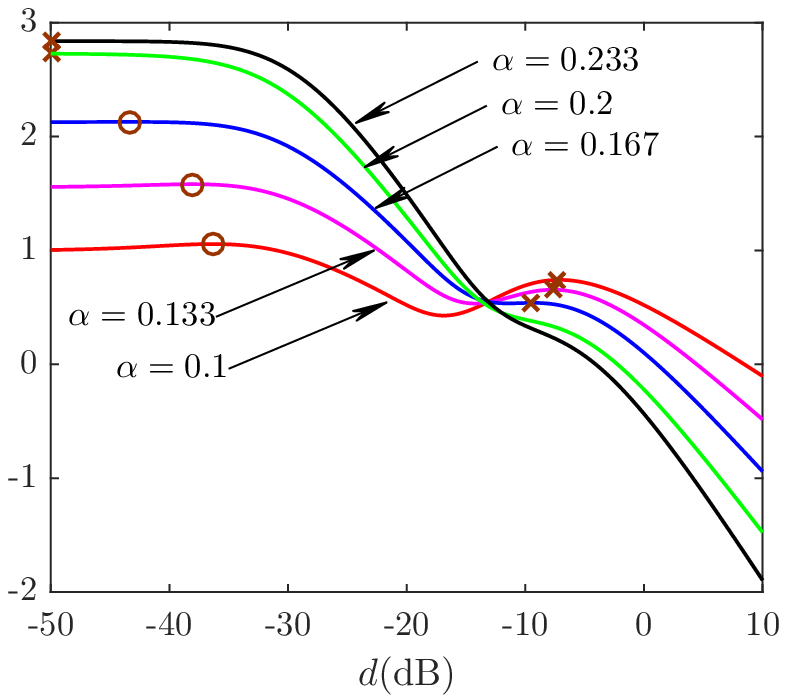}
\label{L2001}
\end{minipage}
}
\subfigure[$L = 2$, $\bar{\sigma}^2 = 0.1$.]{
\begin{minipage}[t]{0.31\linewidth}
\centering
\includegraphics[width=2.3in]{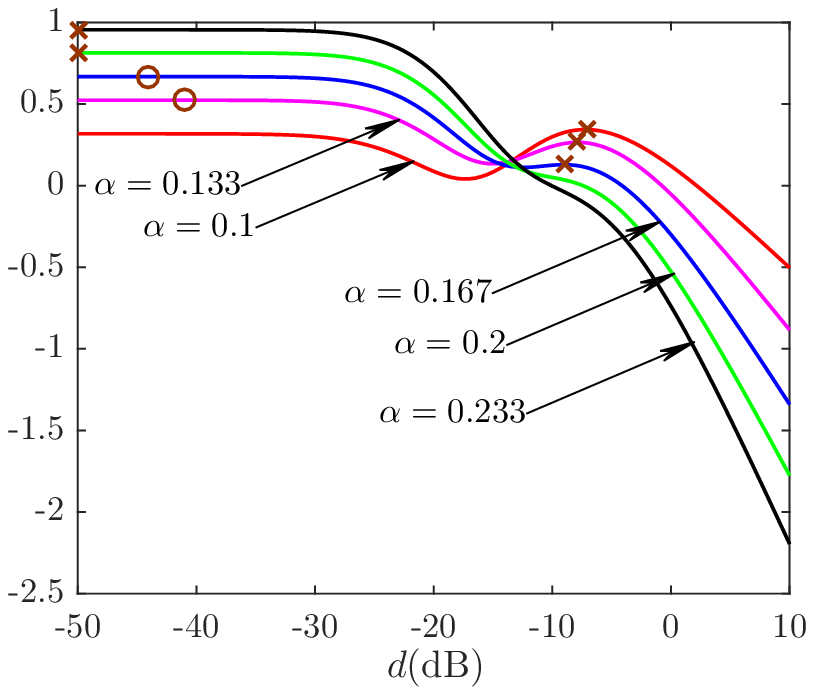}
\label{L201}
\end{minipage}
}
\subfigure[$L = 2$, $\bar{\sigma}^2 = 1$.]{
\begin{minipage}[t]{0.31\linewidth}
\centering
\includegraphics[width=2.3in]{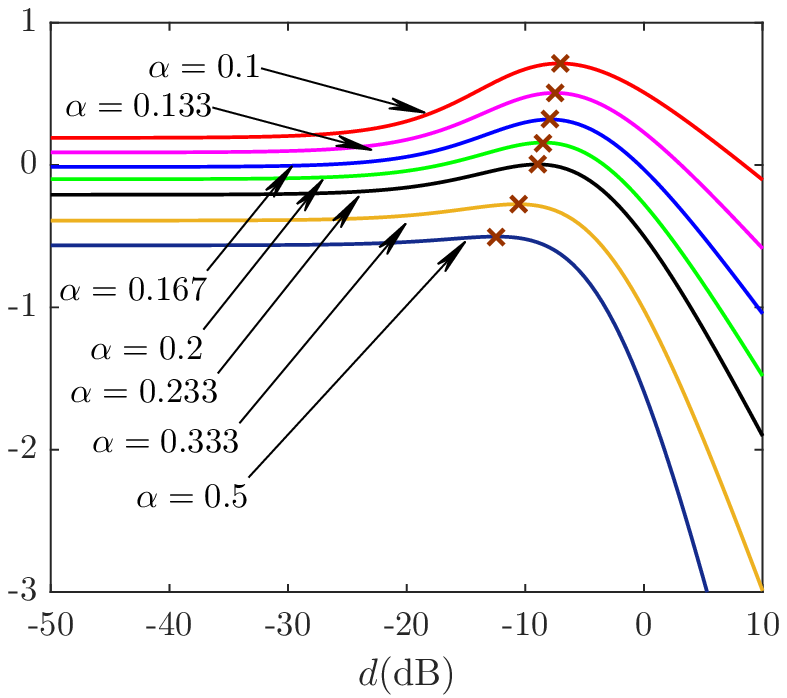}
\label{L21}
\end{minipage}
}
\subfigure[$L = 3$, $\bar{\sigma}^2 = 0.01$.]{
\begin{minipage}[t]{0.31\linewidth}
\centering
\includegraphics[width=2.3in]{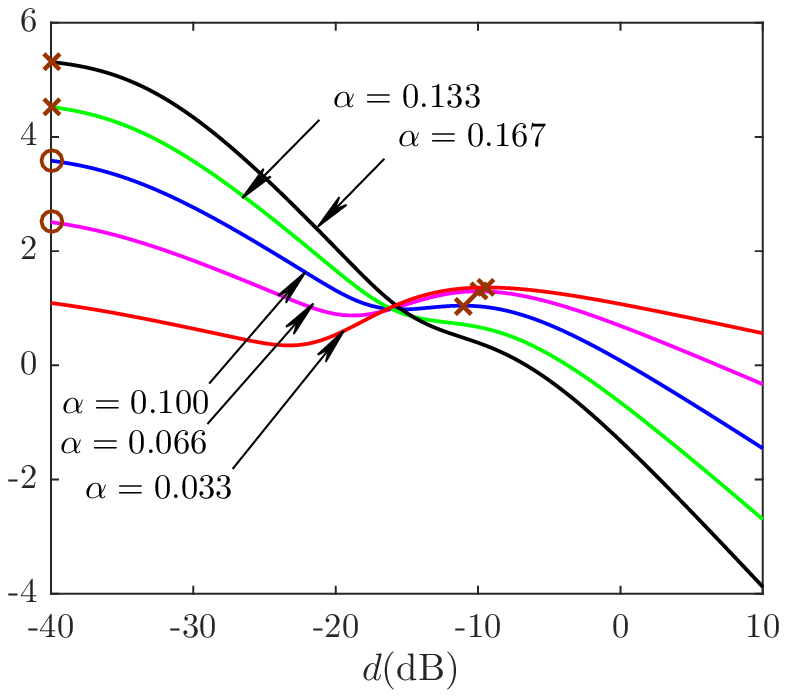}
\label{L3001}
\end{minipage}
}
\subfigure[$L = 3$, $\bar{\sigma}^2 = 0.1$.]{
\begin{minipage}[t]{0.31\linewidth}
\centering
\includegraphics[width=2.3in]{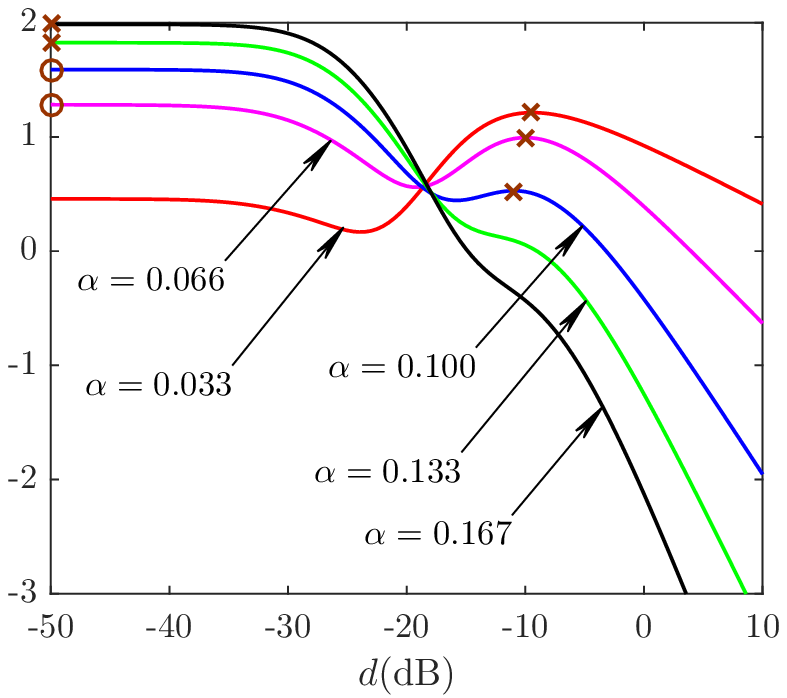}
\label{L301}
\end{minipage}
}
\subfigure[$L = 3$, $\bar{\sigma}^2 = 1$.]{
\begin{minipage}[t]{0.31\linewidth}
\centering
\includegraphics[width=2.3in]{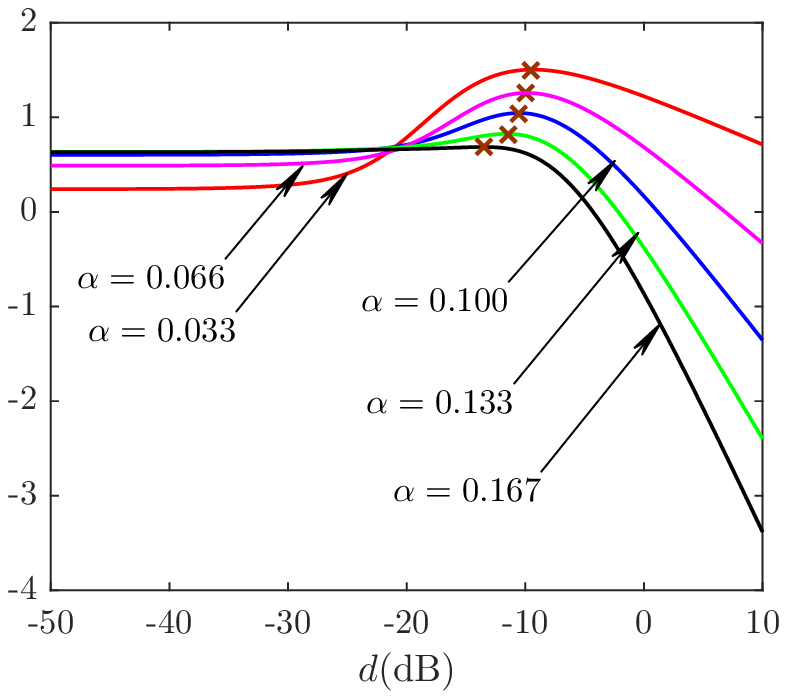}
\label{L31}
\end{minipage}
}
\caption{Free entropy as a function of MSE under different settings of $\bar{\sigma}^2$, $\alpha$ and $L$. (The ``red cross'' and ``red circle'' denote the extreme points of the free entropy function (\ref{phi2}).)}
\label{phase}
\end{figure*}

\subsection{Dealing with the Imperfect CSI}\label{ipcsi}
In previous sections, we have mentioned that the knowledge of CSI is imperfect after the first phase of the proposed URA strategy, which results in an extra inference term in the received signal of the second phase. Accordingly, the received signal model (\ref{equi_mod}) can be reformulated as
\begin{equation}
\mathbf Y = \hat{\mathbf S}\mathbf X+\boldsymbol\Xi+\mathbf I,\label{ipcsi_mod}
\end{equation}
where $\hat{\mathbf S}$ is the estimated channel matrix and  $\mathbf I\triangleq (\mathbf S-\hat{\mathbf S})\mathbf X$ is defined as the interference matrix induced by the channel estimation error. We consider each element in matrix $\mathbf I$ is i.i.d. Gaussian, since we assume the Rayleigh channels for the user devices.

To cope with such a problem, one direct way is to modify our proposed AMP-based algorithm by adding updates to noise variance in each iteration. For conciseness, we will continue to use the previous expressions of $\mathbf S$ and the interference matrix $\mathbf I$ is considered to be absorbed in the noise matrix $\boldsymbol\Xi$. We note that we have dropped the number of iterations in the former derivations for simplification, and we now use $t$ to denote the iteration indexing.

To update the noise variance in each $t$th iteration, we adopt the expectation-maximization (EM) algorithm \cite{EMGAMPSchniter}. In the E-step, we derive the posterior distribution $p(z_{m,j}|y_{m,j};\sigma^2(t))$ with fixed $\sigma^2(t)$,
\begin{align}
&p(z_{m,j}|y_{m,j}; \sigma^2(t))\nonumber \\
&\quad= C^{-1}p(y_{m,j};\sigma^2(t))\mathcal{N}_{\mathcal C}(z_{m,j};\hat{p}_{m,j}(t), Q_{m,j}^p(t))\nonumber\\
&\quad= \mathcal{N}_{\mathcal C}(z_{m,j};\hat z_{m,j}(t),Q_{m,j}^z(t)),
\end{align}
where $C$ is the normalizing constant. In the M-step, we derive the update of $\sigma^2$ by the following equation
\begin{align}
\sigma^2(t+1)&= \arg\max\limits_{\sigma^2>0}\sum\limits_{j = 1}^{2^L}\sum\limits_{m = 1}^M \int\nolimits_{z_{m,j}}p(z_{m,j}|y_{m,j}; \sigma^2(t))\nonumber\\
&\quad\quad\times\ln \mathcal{N}_{\mathcal C}(z_{m,j};y_{j,m},\sigma^2).
\end{align}
As a consequence, the update equation of the noise variance can be specified as
\begin{equation}
\sigma^2(t+1) = \frac{1}{2^LM}\sum\limits_{j = 1}^{2^L}\sum\limits_{m = 1}^M\left(|y_{m,j}-\hat z_{m,j}(t)|^2+Q_{m,j}^z(t)\right).\label{noi_up}
\end{equation}
After that, we will modify our proposed algorithm by adding the noise update equation (\ref{noi_up}) in each iteration. The concrete implementing schedule and details of the proposed algorithm are specified in the Algorithm \ref{alg1}. We note that in each sub-block of the second phase, we first implementing a matrix multiplication operation of complexity $\mathcal O(2^{2L}M)$, then we executing the Algorithm \ref{alg1} of complexity $\mathcal O(2^{L}KM)$ in each iteration. Since it is always the case that $K\gg 2^L$, the overall computational complexity of each iteration in the second phase is $\mathcal O((B-L_0)L^{-1}2^{L}KM)$. We further notice that the overall complexity in each iteration in the first phase is $\mathcal O(2^{L_0}nM)$, which equals the complexity in each sub-block of the divide-and-conquer strategies, such as \cite{Shyianov2021}. As aforementioned, we need $n$ in the same scale of $K$, and we have $(BL_0^{-1}-1)2^{L_0}> (B-L_0)L^{-1}2^{L}$ with $L_0> L$. Hence, the overall complexity of our propose URA scheme is much less.

\begin{figure*}
\begin{center}
\subfigure[$L = 2$.]{
\begin{minipage}[t]{0.31\linewidth}
\centering
\includegraphics[width=2.3in]{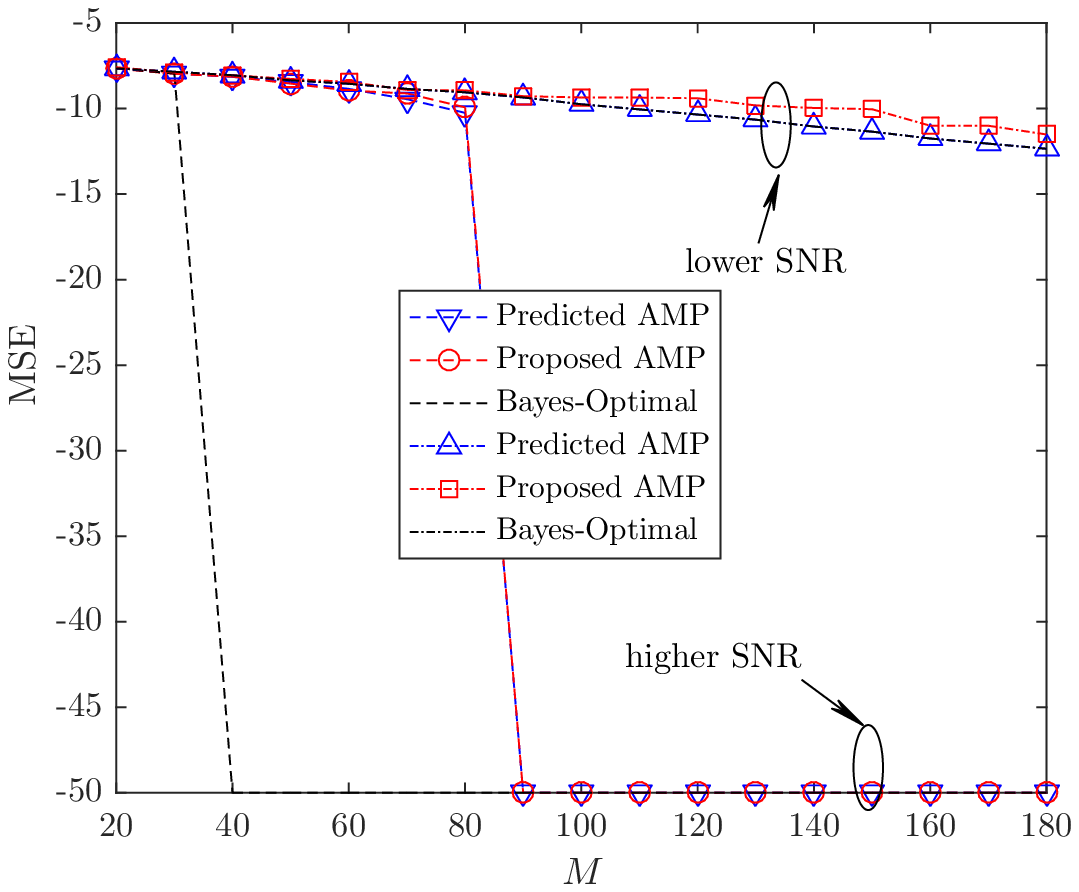}
\label{preL2}
\end{minipage}%
}
\subfigure[$L = 3$.]{
\begin{minipage}[t]{0.31\linewidth}
\centering
\includegraphics[width=2.3in]{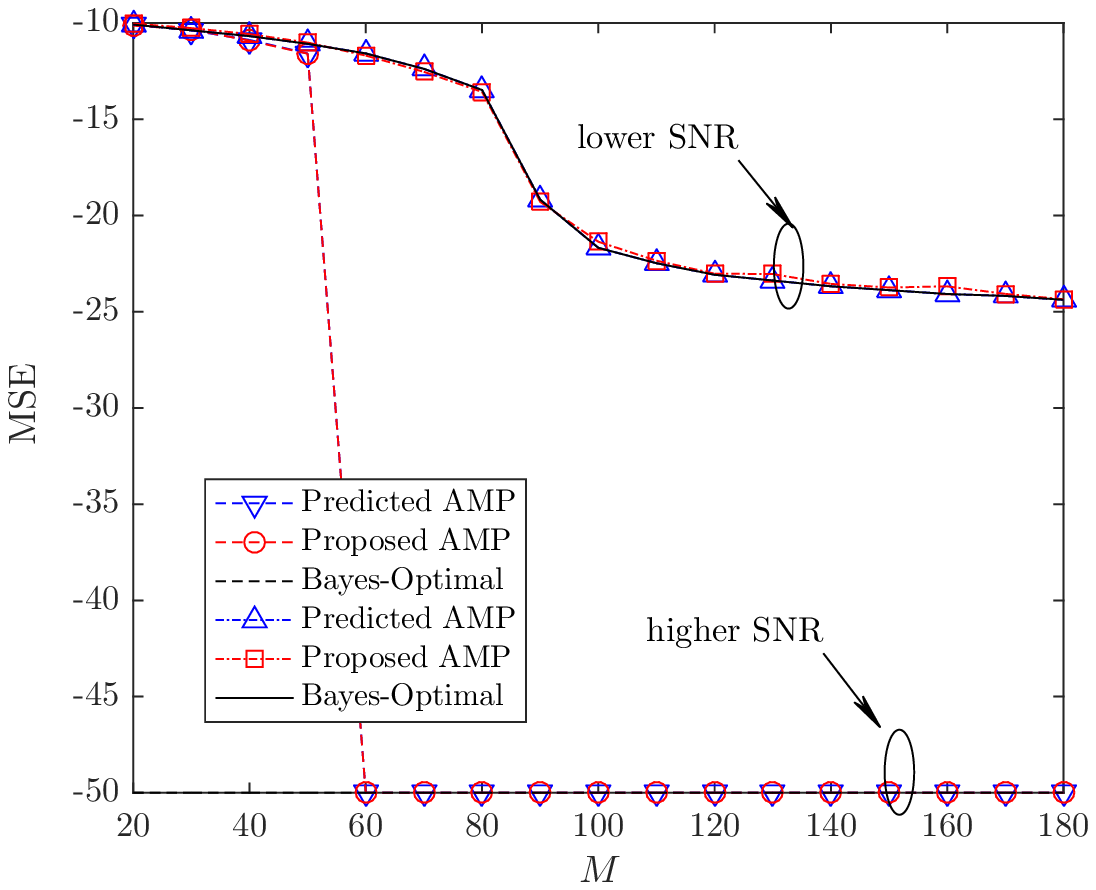}
\label{preL3}
\end{minipage}
}
\subfigure[$L = 4$.]{
\begin{minipage}[t]{0.31\linewidth}
\centering
\includegraphics[width=2.3in]{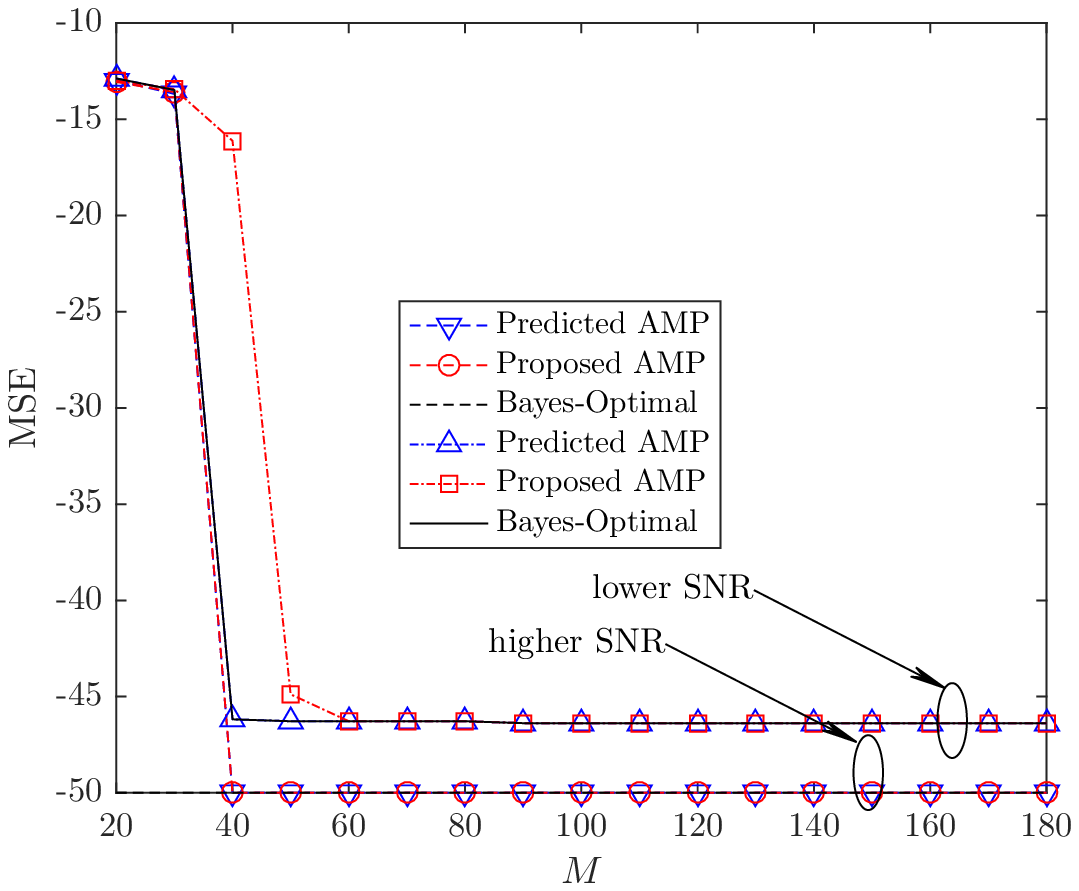}
\label{preL4}
\end{minipage}
}
\caption{MSE performances with different settings of $L$ versus the number of antennas $M$.}
\label{preL}
\end{center}
\end{figure*}


\section{Numerical Result}\label{num_res}
In this section, we verify the performance of the proposed URA strategy with several numerical results. We consider a Raleigh fading channel with unit power and we assume $\rho_k = 1, \forall k$ and the power factors are absorbed into the noise variances. For clarification, the variances of normalized additive noise in the first and second phases are denoted as $\tilde{\sigma}^2$ and $\bar{\sigma}^2$, respectively. In the following, both the verifications of the performance analysis and proposed decoding algorithm are provided. We also demonstrate the performances in the practical settings. Some comparisons are made with the divide-and-conquer counterparts in \cite{2019ale, Shyianov2021} and the pilot-based counterparts in \cite{PilotFengler, PilotAhmadi, FASURA}.

\subsection{Verifications of the Performance Analysis Results}
In this section, we investigate the performance analysis results for the proposed URA scheme, where we suppose the CSI error in the first phase is absorbed in the noise variance $\bar{\sigma}^2$ of the second phase. Concretely, we demonstrate the free entropy function given in (\ref{phi2}) with different system parameters, i.e., $\bar{\sigma}^2$, $\alpha$ and $L$. As aforementioned, the \emph{Bayes-optimal MSE} solution is determined by the extreme point of function (\ref{phi2}) associated with the largest value of $\Phi_2$, and the AMP algorithm always converge to the extreme point of function (\ref{phi2}) associated with the largest MSE, referred as \emph{AMP achievable MSE}. Based on that, we will show that the MSE performance reflected on the extreme points of (\ref{phi2}) can be divided into some regions.

Fig. \ref{phase} demonstrates the fixed point of free entropy function (\ref{phi2}), which reflects the MSE performance under different system parameters. We can observe from Fig. \ref{phase} that there are phase transitions with the variations of the parameter $\alpha$ under different values of $\bar{\sigma}^2$ and $L$, dividing the MSE performance into some particular regions. The phase transitions means the changes of the number of extreme points of the free entropy function, leading to that the performances of the decoding algorithm changes dramatically to achieve perfect recovery. Taking the Fig. {\ref{L201}} as an example, we can infer that there exists some hard thresholds that divides the MSE performance into regions. Concretely, there exists an $\alpha_1$ within the range $(0.1, 0.133)$, so that in the region $\alpha\in(0,\alpha_1)$, the Bayes-optimal MSE coincides the AMP achievable MSE. Within the range $(0.167, 0.2)$, there exists an $\alpha_2$, so that in the region $\alpha\in(\alpha_1,\alpha_2)$, there exists two extreme points of the function (\ref{phi2}) and the smaller extreme point leads to a larger value of the free entropy function, indicating the Bayes-optimal MSE. However, the AMP achievable MSE is blocked by another extreme point associated with the larger extreme point, so that there exists a performance gap between the AMP achievable MSE and Bayes-optimal MSE. Finally, in the region $\alpha\in(\alpha_2,\infty)$, the larger extreme point disappears, and the AMP achievable MSE and Bayes-optimal MSE coincides again. We note that the phase transition occurs in $\alpha = \alpha_2$, where there is a sudden change in the number of extreme points and in the value of the AMP achievable MSE.

We can further observe from Fig. \ref{phase} that the phase transition gradually disappears with increasing of the noise variance. When the value of $\bar{\sigma}^2 = 1$, the phase transitions in all the settings of $L$ disappear, and the MSE performance can be continuously improved by increasing the number of antennas with fixed number of user devices. On the other hand, the Fig. \ref{phase} also shows that the value of $\alpha_2$ where the phase transition occurs, will be reduced when the length of section $L$ increases. Noticing the cases with $\bar{\sigma}^2 = 0.1$, we obtain $\alpha_2\in(0.233,0.267)$ for the case of $L = 1$, $\alpha_2\in(0.167,0.2)$ for the case of $L = 2$ and $\alpha_2\in(0.1,0.133)$ for the case of $L = 3$. Such an observation is explainable, since with the increasing of $L$ the probability of the codeword collisions between user devices will be reduced, resulting in the enhancement of sparsity over each column of matrix $\mathbf X$. The above phase transition analysis provides the essential reasons of the URA performances of our pattern. In addition, such an analysis also explains the typical phenomenon that with a given CSI, we can support more user devices by increasing the number of antennas. This is because the performance is reflected by the ratio $\alpha$ between the number of antennas and user devices in a large system regime.

Fig. \ref{preL} investigates the MSE performances with different settings of parameter $L$ with respect to the number of antennas $M$, i.e., the Bayes-optimal MSE, AMP achievable MSE predicted by the free entropy function, referred as \emph{predicted AMP} and MSE obtained by implementing the proposed AMP algorithm in sec. \ref{pro_amp}, referred as \emph{proposed AMP}. The number of active users is set as $K = 500$. The notations \emph{lower SNR} and \emph{higher SNR} indicate the case with $\bar{\sigma}^2 = 0.01$ and $\bar{\sigma}^2 = 1$, respectively. We can observe from Fig. \ref{preL} that the MSE performance of our proposed AMP algorithm matches the MSE performance predicted by the replicated free entropy function well in both the cases of lower SNR and higher SNR. In addition, when the number of antennas is large enough, the MSE performance of the proposed AMP is Bayes-optimal, which is exactly the case that there exists only one extreme point of replicated free entropy function. Obviously, we can observe the phase transition phenomenon and the performance gap between Bayes-optimal MSE and AMP achievable MSE in the case of higher SNR. Particularly, the number of antennas required for occurring phase transition is reduced with the increase of $L$, corresponding to the results mentioned before. Another noteworthy phenomenon is that the recovery MSE decline rapidly with $M$ increases with larger $L$, even in the case of lower SNR. This implies that increasing the length of section is a workable manner to cope with the case of low SNR with massive MIMO.
Fig. \ref{preL} also implies an important result in the context of MIMO URA. With an enough energy, i.e., $\bar{\sigma}^2 = 0.01$, the BS can serve a large number of user devices, i.e., $K = 500$ with a very low $L = 2$ by providing a large number of antennas, i.e., $M = 90$. This means the overall length of the codewords in the data phase approximately linearly increases, as the increase of the number of the information bits, i.e., the overall length of the codewords in the data phase equals $2B$ with $B$ bits of information for the case of $L = 2$. This phenomenon indicates a significant advantage of the two-phase scheme compared with the existing divide-and-conquer counterparts, since they all map a piece of data with each sub-block to a long transmitted codeword, which is directly proportional to the number of the active users.
\begin{figure}
\subfigure[$n = 500$.]{
\begin{minipage}[t]{0.5\linewidth}
\centering
\includegraphics[width=3.5in]{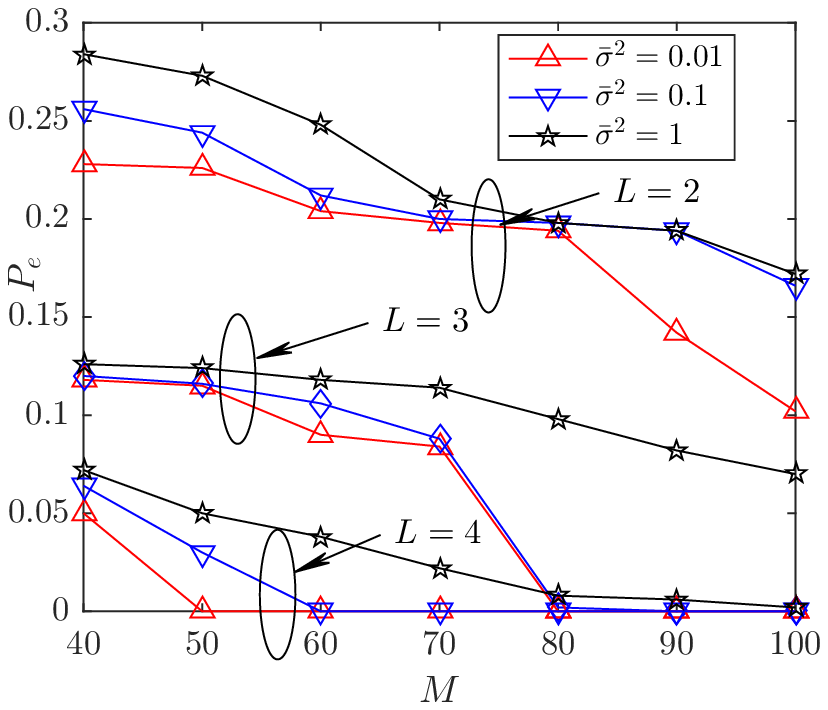}
\label{pen500}
\end{minipage}%
}
\subfigure[$n = 1000$.]{
\begin{minipage}[t]{0.5\linewidth}
\centering
\includegraphics[width=3.5in]{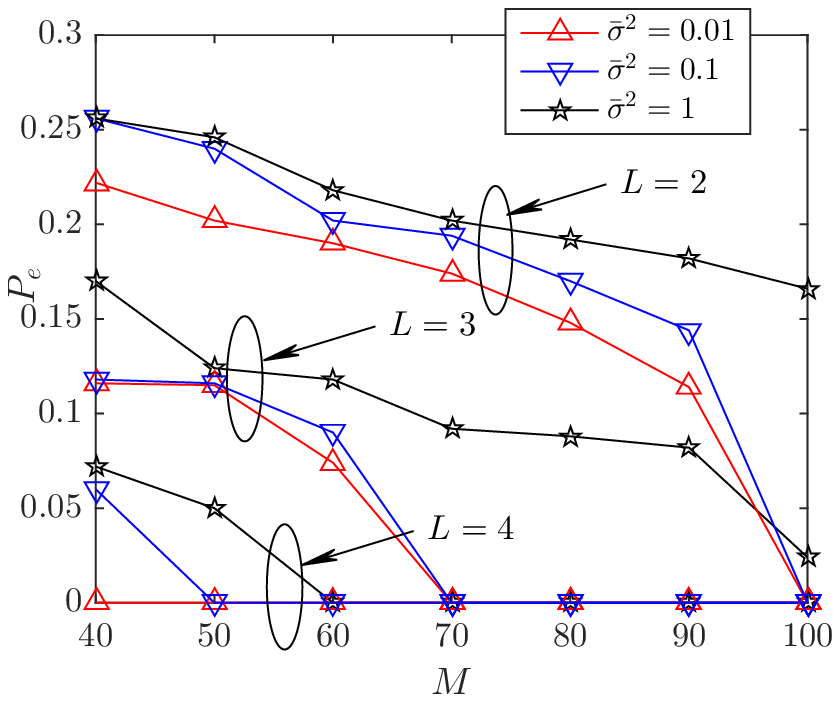}
\label{pen1000}
\end{minipage}
}
\caption{Error probability of our proposed URA strategy with different settings of $L$ versus the number of antennas $M$ in particular settings of system parameters in the first phase.}
\label{pepro}
\end{figure}

\subsection{Simulations in Practical Settings}
In this section, we verify the proposed two-phase URA strategy in the practical settings. We suppose the total $K = 500$ active user devices transmit information sequences of the same length $B = 100$. In the first phase, each user transmit a sub-block of length $L_0 = 16$. The codeword sent by each user is from a codebook consisting of $2^{L_0}$ components of length $n$, whose entries are generated from $\mathcal{CN}(0,1/n)$. To jointly estimate channel and recover messages, the standard AMP algorithm in \cite{Shyianov2021} is adopted. We note that we in this paper assume that there are no user codeword collisions, which can be guaranteed by adopting the collision protocol introduced in \cite{2022Li}. In the second phase, the residual messages for each user are partitioned into $J$ sub-blocks of equal length $L$ with $B-L_0 = LJ$. Our proposed AMP-based method with noise variance tuning is implemented based on the estimated channel. Note that the notation of the noise variance $\bar{\sigma}^2$ will not include the channel estimation error in the first phase.

Fig. \ref{pepro} investigates the error probability of our proposed URA strategy. Note that in the first phase, we set the parameter $\tilde{\sigma}^2 = 0.01$, and the length of codewords $n = 500$ and $n = 1000$. We can observe the phase transition phenomenon from Fig. \ref{pepro}, i.e., the error probability reduces to zero cross the threshold of some system parameters. We can intuitively notice that increasing the number of antennas, increasing the length of sub-blocks and increasing the system SNR benefit the performance of our scheme to achieve the perfect recovery. In addition, the increasing channel estimation error in the first phase increases the required number of antennas to achieve the perfect recovery in the second phase. Concretely, in the case of $n = 500$, the phase transition never occurs in the settings of $L = 2$ and $\bar{\sigma}^2 = 1$, even with $M = 100$ antennas. In other cases, $M = 80$ is enough for perfect recovery of information sequences. In the case of $n = 1000$, $M = 100$ is enough for nearly all the considered settings of system parameters, expected to the case of $\bar{\sigma}^2 = 1$ and $L = 2, 3$.
\begin{figure}
\centering
\includegraphics[width=3.5in]{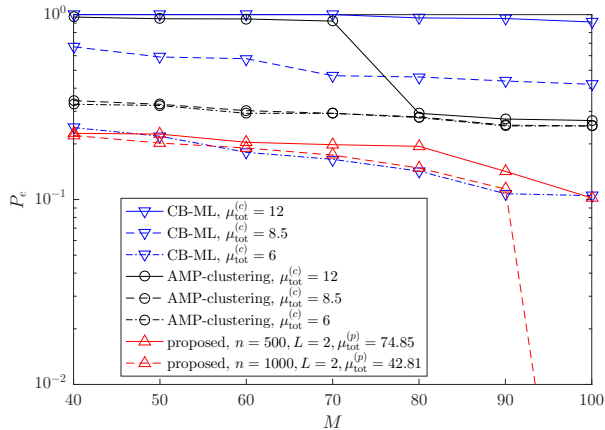}
\caption{The error probability of our URA scheme and the counterpart strategies versus the number of antennas $M$ with different settings of spectral efficiency. }
\label{comp}
\end{figure}

For demonstrating the strengths of our proposed URA strategy, we first make a comparison with the state-of-the-art divide-and-conquer counterparts proposed in \cite{2019ale, Shyianov2021}, which are the most recent strategies in our considered scenario to the best of our knowledge, referred as \emph{CB-ML} and \emph{AMP-clustering}, respectively. Concretely, the length of the information bits is set as $B = 96$ for all of them. The length of the bits in each sub-block is set as $12$. The number of the sub-blocks is denoted as $\bar{J}$. For CB-ML, the code rate is typically set as $0.25$, the number $\bar{J}$ therefore equals $32$. For AMP-clustering, we have $\bar{J} = 8$. In order to differentiate, for the counterpart strategies, we denote that the codeword in each sub-block sent by each user is of length $n_c$. We then define the total spectral efficiency as $\mu_{\rm tot} = KB/T_{\rm tot}$, where $T_{\rm tot}$ is defined as the total length of the transmitted symbols of the user devices in a transmission. For our proposed URA strategy, the total spectral efficiency is
$
\mu_{\rm tot}^{(p)} = \frac{KB}{n+(B-L_0)2^L/L},\nonumber
$
while the total spectral efficiency for the counterparts is $\mu_{\rm tot}^{(c)} = {KB}/{n_c\bar{J}}$.

Fig. \ref{comp} demonstrates the error probability of our URA scheme and the counterpart strategies versus the number of antennas $M$ with different settings of spectral efficiencies. We set the normalized noise variances as $0.01$.  Compared with these two counterparts, our proposed scheme has significant strengths, as seen in Fig. \ref{comp}. Even with $\mu_{\rm tot}^{(c)} = 6$ spectral efficiency, the divide-and-conquer counterparts can not achieve the low error probability performance. On the contrary, for our proposed algorithm, the length of codewords is set as $2^L = 4$ of each sub-block of the second phase. This results in a very high spectral efficiency. Specifically, we can observe that with $\mu_{\rm tot}^{(p)} = 42.81$ spectral efficiency, our proposed algorithm can achieve a very low error probability with $M = 100$. This result demonstrates that the two-phase URA scheme has a great strength compared with the divide-and-conquer counterparts in term of the spectral efficiency.

\begin{figure}[!t]
\centering
\includegraphics[width=3.5in]{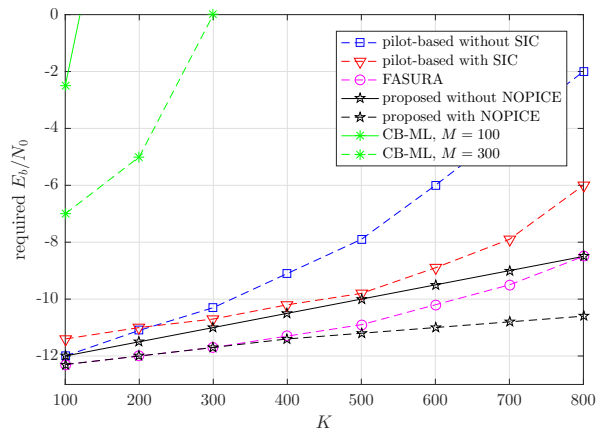}
\caption{$E_b/N_0$ versus $K$ curves of proposed algorithm and the typical baselines.}
\label{ebn0}
\end{figure}
To further demonstrate the energy efficiency of the proposed URA scheme, we compare our proposed algorithm with typical baselines by a series of $E_b/N_0$ versus $K$ curves, where the $E_b/N_0$ in dB is the energy-per-bit of the system. We consider the typical CB-ML algorithm as one of the baselines in the context of divide-and-conquer algorithms. We also choose three typical pilot-based baselines in \cite{PilotFengler, PilotAhmadi, FASURA}. The method in \cite{PilotFengler} is referred as \emph{pilot-based method without SIC}. The method in \cite{PilotAhmadi} is referred as \emph{pilot-based method with SIC}. The method in \cite{FASURA} is referred as \emph{FASURA}, which combined the pilot-based method with an additional operation NOPICE. The length of the overall channel uses is set as $3200$.  The power ratio between the pilot and the data phase is set as $1$. The error probability threshold is $0.05$. The total number of the information bits is set as $B = 100$. The number of the BS antennas is set as $M = 50$, unless additional notes are made.

\setcounter{TempEqCnt}{\value{equation}}
\setcounter{equation}{32}
\begin{figure*}[tbp]
\begin{align}
1 &= \int\left[\prod_{a = 1}^n{\rm d}\mathbf Q^{(a)}{\rm d}\hat{\mathbf Q}^{(a)}{\rm d}\mathbf M^{(a)}{\rm d}\hat{\mathbf M}^{(a)}\right]\left[\prod_{b,a<b}^{n(n-1)/2}{\rm d}\mathbf T^{(ab)}{\rm d}\hat{\mathbf T}^{(ab)}\right]\nonumber\\
 &\times \exp\left({\rm Tr}\left(-\sum\limits_{a = 1}^n\hat{\mathbf M}^{(a)T}\left(K\hat{\mathbf M}^{(a)}-\sum\limits_{k = 1}^K\mathbf x_{k}^T\hat{\mathbf x}_{k}^{(a)}\right)+\sum\limits_{a = 1}^n\hat{\mathbf Q}^{(a)T}\left(\frac{K}{2}\hat{\mathbf Q}^{(a)}-\frac{1}{2}\sum\limits_{k = 1}^K\mathbf x_{k}^{(a)T}\hat{\mathbf x}_{k}^{(a)}\right)\right.\right.\nonumber\\
 &\left.\left.-\sum\limits_{b,a<b}^{n(n-1)/2}\hat{\mathbf T}^{(ab)T}\left(K\hat{\mathbf T}^{(ab)}-\sum\limits_{k = 1}^K\mathbf x_{k}^{(a)T}\hat{\mathbf x}_{k}^{(b)}\right)\right)\right).\label{iden}
\end{align}
\begin{align}
&\mathbb E_{\mathbf S,\mathbf Y}(\mathcal Z(\mathbf Y, \mathbf S)^n) = \int\left[\prod_{a = 1}^n{\rm d}\mathbf Q^{(a)}{\rm d}\hat{\mathbf Q}^{(a)}{\rm d}\mathbf M^{(a)}{\rm d}\hat{\mathbf M}^{(a)}\right]\left[\prod_{b,a<b}^{n(n-1)/2}{\rm d}\mathbf T^{(ab)}{\rm d}\hat{\mathbf T}^{(ab)}\right]\prod_{m = 1}^M\left|\mathbf I_{2^Ln}+\sigma^{-2}\mathbf G\right|^{-1}\nonumber\\
\times&\exp\left(K{\rm Tr}\left(-\sum\limits_{a = 1}^n\hat{\mathbf M}^{(a)T}\mathbf M^{(a)}+\sum\limits_{a = 1}^n\frac{1}{2}\hat{\mathbf Q}^{(a)T}\mathbf Q^{(a)}-\sum\limits_{b,a<b}^{n(n-1)/2}\hat{\mathbf T}^{(ab)T}\mathbf T^{(ab)}\right)\right)\left(\sum\limits_{\mathbf x,\hat{\mathbf x}^{(1)},\dots,\hat{\mathbf x}^{(n)}}p(\mathbf x)\right.\nonumber\\
\times&\left.\prod_{a = 1}^np(\hat{\mathbf x}^{(a)})\exp\left({\rm Tr}\left(\sum\limits_{a = 1}^n\hat{\mathbf M}^{(a)T}\mathbf x^T\hat{\mathbf x}^{(a)}-\frac{1}{2}\sum\limits_{a = 1}^n\hat{\mathbf Q}^{(a)T}\hat{\mathbf x}^{(a)T}\hat{\mathbf
x}^{(a)}+\sum\limits_{b,a<b}^{n(n-1)/2}\hat{\mathbf T}^{(ab)T}\hat{\mathbf x}^{(a)T}\hat{\mathbf x}^{(b)}\right)\right)\right)^K.\label{ave_rep_par}
\end{align}
\setcounter{equation}{36}
\begin{align}
&\mathcal X_m= \left|\mathbf I_{2^Ln}+\sigma^{-2}\mathbf G\right| \nonumber\\
&\overset{\lim\limits_{n\to 0}}{\approx}\left(\mathbf I_{2^Ln}+n{\rm Tr}\left(\left(\boldsymbol\Delta+\frac{K}{M}(\mathbf Q-\mathbf T)\right)^{-1}\left(\frac{K}{M}(\mathbf C+\mathbf T-\mathbf M-\mathbf M^T)+\boldsymbol\Delta\right)\right)\right)\left|\mathbf I_{2^Ln}+\frac{K(\mathbf Q-\mathbf T)}{M\sigma^2}\right|^n\nonumber\\
&\overset{\lim\limits_{n\to 0}}{\approx}\exp\left(n{\rm Tr}\left(\left(\boldsymbol\Delta+\frac{K}{M}(\mathbf Q-\mathbf T)\right)^{-1}\left(\frac{K}{M}(\mathbf C+\mathbf T-\mathbf M-\mathbf M^T)+\boldsymbol\Delta\right)\right)+\log\left|\mathbf I_{2^Ln}+\frac{K(\mathbf Q-\mathbf T)}{M\sigma^2}\right|\right).\label{xm2}
\end{align}
\setcounter{equation}{38}
\begin{align}
&\tilde{\Phi}\left(\mathbf M, \mathbf Q, \mathbf T, \hat{\mathbf M}, \hat{\mathbf Q}, \hat{\mathbf T}\right)\nonumber\\
&= {\rm Tr}\left(-\hat{\mathbf M}^T\mathbf M+\frac{1}{2}\hat{\mathbf Q}^T\mathbf Q+\frac{1}{2}\hat{\mathbf T}^T\mathbf T\right)-\alpha{\rm Tr}\left(\left(\boldsymbol\Delta+\alpha^{-1}(\mathbf Q-\mathbf T)\right)^{-1}\left(\alpha^{-1}(\mathbf C+\mathbf T-\mathbf M-\mathbf M^T)+\boldsymbol\Delta\right)\right)\nonumber\\
&-\alpha\log\left|\mathbf I_{2^Ln}+\frac{\mathbf Q-\mathbf T}{\alpha\sigma^2}\right|+\sum\limits_{\mathbf x}p(\mathbf x)\int D\mathbf z\log\left(\sum\limits_{\hat{\mathbf x}}p(\hat{\mathbf x})\exp\left({\rm Tr}\left(\hat{\mathbf M}^T\mathbf x^T\hat{\mathbf x}-\frac{1}{2}\hat{\mathbf Q}^T\hat{\mathbf x}^{T}\hat{\mathbf x}\right.\right.\right.\nonumber\\
&\left.\left.\left.-\hat{\mathbf T}^{\frac{1}{2}}\sqrt{2}\mathfrak R(\mathbf z^T)\hat{\mathbf x}-\frac{1}{2}\hat{\mathbf T}\hat{\mathbf x}^{T}\hat{\mathbf x}\right)\right)\right).\label{tid_phi}
\end{align}
\hrulefill
\end{figure*}
\setcounter{equation}{\value{TempEqCnt}}
For our proposed algorithm, we set two groups, which are both assigned $1600$ channel uses. For each group, we set $n = 704$ in the first phase, and we set $L = 6$ and provide $896$ channel uses for the second phase. Each user randomly chooses its group. Fig. \ref{ebn0} demonstrates the $E_b/N_0$ versus $K$ curves of proposed algorithm and the typical baselines. We can observe that even with $M = 100$ or $M = 300$, there is still a large performance gap between the CB-ML algorithm and the other two-phase algorithms in the quasi-static channel. We can also observe that our proposed algorithm achieves a significant better performance compared with both pilot-based method without SIC and the pilot-based method with SIC. This is because our algorithm is based on the MMSE criterion, which considers the prior distribution (\ref{con_pri}) and the channel estimation error to recover the index matrix/information bits. On the contrary, the baselines use a simple maximal-ratio combiner (MRC) or linear MMSE to recover the codewords, and use a list decoder to recover the information bits. Such a operation causes a loss of information because they do not utilize the prior knowledge of the coding structure and channel estimation error in the recovery of the codewords. In addition, we can see that the FASURA improves the performance of the pilot-based method by the NOPICE. For fair comparison, we add an additional NOPICE in the same way in our proposed algorithm, referred as \emph{proposed with NOPICE}. This can be done by reconstructing the codewords based on the estimated index matrix in each sub-block. After that, we can see that the proposed algorithm with NOPICE achieves a better energy efficiency than the FASURA. This result shows the great advantage of our scheme in term of the energy efficiency.

\section{Conclusion}

This paper proposed a novel two-phase URA strategy. In the first phase, the CSI knowledge was acquired. By adopting the CSI, the sub-block recovery problem in the second phase can be formulated as a specific CS recovery problem. We theoretically showed that the codeword collisions caused by the short length of sub-blocks can be coped with by equipping the enough number of the BS antennas. For the fixed number of active users, the larger number of antennas allows us to design shorter lengths of sub-blocks with the given requirement of error probability. Based on our theoretical framework, we established the connection between the decoding performances of our URA scheme and the system parameters. The phase transition phenomenons were also shown to reveal the essential reasons of the decoding performances. We designed the AMP-based decoder for the second phase with noise variance tuning to cope with the channel estimation error. By adopting our proposed URA strategy with massive MIMO, we achieved the performance improvement compared with the most recent counterparts.


\appendix
\subsection{Proof of Theorem 1}\label{ap1}
The replicated partition function can be calculated via
\begin{align}
\mathbb E_{\mathbf S,\mathbf Y}&(\mathcal Z(\mathbf Y, \mathbf S)^n)\nonumber\\
&= \sum\limits_{\mathbf X, \hat{\mathbf X}^{(1)},\dots, \hat{\mathbf X}^{(n)}}p(\mathbf X)\prod_{a = 1}^np(\hat{\mathbf X}^{(a)})\prod_{m = 1}^M\mathcal X_m,\label{EZn}
\end{align}
where
\begin{align}
\mathcal X_m = & \mathbb E_{\mathbf S,\boldsymbol\Xi}\left(\prod_{a = 1}^n\exp\left(-\sigma^{-2}\left(\sum\limits_{k = 1}^Ks_{m,k}(\mathbf x_k-\hat{\mathbf x}_{k}^{(a)})+\boldsymbol\xi_k\right)\right.\right.\nonumber\\
&\quad\quad\quad\times\left.\left.\left(\sum\limits_{k = 1}^Ks_{m,k}(\mathbf x_k-\hat{\mathbf x}_{k}^{(a)})+\boldsymbol\xi_k\right)^H\right)\right).\nonumber
\end{align}
In the above equations, $a = 1,\dots, n$ denotes the replica indices, $p(\mathbf X)$ is the true prior distribution over signal $\mathbf X$ and $\boldsymbol\xi_k, k = 1,\dots, K$ is the row of the noise matrix $\boldsymbol\Xi$. To calculate the term $\mathcal X_m$ in (\ref{EZn}), we define a random vector $\mathbf v_m\triangleq [\boldsymbol\nu_{m}^{(1)}, \dots, \boldsymbol\nu_{m}^{(n)}]\in\mathbb C^{1\times 2^Ln}$ with each element $\boldsymbol\nu_{m}^{(a)} = K^{-1}\sum\nolimits_{k = 1}^Ks_{m,k}(\mathbf x_k-\hat{\mathbf x}_{k}^{(a)})+\boldsymbol\xi_k$. Using the fact each element of the matrix $\mathbf S$ as well as $\boldsymbol\Xi$ is i.i.d. Gaussian random variable, the $\mathbf v_m$ obeys a joint distribution with the following moments formulated as
\begin{align}
&\mathbb E_{\mathbf v_m}(\boldsymbol\nu_{m}^{(a)H}\boldsymbol\nu_{m}^{(a)}) = \frac{1}{M}\sum\nolimits_{k = 1}^K(\mathbf x_k-\hat{\mathbf x}_{k}^{(a)})^T(\mathbf x_k-\hat{\mathbf x}_{k}^{(a)})+\boldsymbol\Delta.\nonumber\\
&\mathbb E_{\mathbf v_m}(\boldsymbol\nu_{m}^{(a)H}\boldsymbol\nu_{m}^{(b)}) = \frac{1}{M}\sum\nolimits_{k = 1}^K(\mathbf x_k-\hat{\mathbf x}_{k}^{(a)})^T(\mathbf x_k-\hat{\mathbf x}_{k}^{(b)})+\boldsymbol\Delta.\nonumber
\end{align}
Now we define some new parameters, referred as overlaps:
\begin{align}
\mathbf C&\triangleq K^{-1}\sum\nolimits_{k = 1}^K\mathbf x_k^T\mathbf x_k, \mathbf Q^{(a)} \triangleq K^{-1}\sum\nolimits_{k = 1}^K\hat{\mathbf x}_{k}^{(a)T}\hat{\mathbf x}_{k}^{(a)},\nonumber\\
\mathbf M^{(a)} &\triangleq  K^{-1}\sum\nolimits_{k = 1}^K\mathbf x_{k}^T\hat{\mathbf x}_{k}^{(a)}, \mathbf T^{(ab)} \triangleq K^{-1}\sum\nolimits_{k = 1}^K\hat{\mathbf x}_{k}^{(a)T}\hat{\mathbf x}_{k}^{(b)}.\label{ove}
\end{align}
We note the integrand depends on the $\mathbf X$ and its replicas only through the considering overlaps, and
\begin{align}
&\mathbb E_{\mathbf v_m}(\boldsymbol\nu_{m}^{(a)H}\boldsymbol\nu_{m}^{(a)}) = \frac{K}{M}\left(\mathbf C-\mathbf M^{(a)}-\mathbf M^{(a)T}+\mathbf Q^{(a)}\right)+\boldsymbol\Delta,\nonumber\\
&\mathbb E_{\mathbf v_m}(\boldsymbol\nu_{m}^{(a)H}\boldsymbol\nu_{m}^{(b)}) = \frac{K}{M}\left(\mathbf C-\mathbf M^{(a)}-\mathbf M^{(b)T}+\mathbf T^{(ab)}\right)+\boldsymbol\Delta.\nonumber
\end{align}
The covariance matrix of $\bf{v}_m$ is denoted as $\mathbf G\in\mathbb C^{2^Ln\times 2^Ln}$, which is independent of index $m$, reads that $\forall a,b$, the corresponding block $\mathbf G^{(ab)}\in\mathbb C^{2^L\times 2^L} = \mathbb E_{\mathbf v_m}(\boldsymbol\nu_{m}^{(a)H}\boldsymbol\nu_{m}^{(b)})$. Accordingly, we obtain
\begin{align}
\mathcal X_m = \int{\rm d}\mathbf vp(\mathbf v)\exp(-\sigma^{-2}\sum\limits_{a = 1}^n||\boldsymbol\nu_{m}^{(a)}||^2)= \left|\mathbf I_{2^Ln}+\sigma^{-2}\mathbf G\right|^{-1}.\label{Xm}
\end{align}
To enforce the overlaps in (\ref{EZn}) to satisfy their definitions (\ref{ove}), we plug an identity based on the inverse Fourier transform of the Dirac delta function, given by (\ref{iden}). The newly defined parameters $\{\hat{\mathbf Q}^{(a)}, \hat{\mathbf M}^{(a)}, \hat{\mathbf T}^{(ab)}\}$ are called conjugated parameters, which enforcing the consistency conditions of (\ref{ove}). We note that the identity in (\ref{iden}) is obtained heuristically from that in \cite{Barbier2017}. Plugging (\ref{iden}) into (\ref{EZn}) and combining (\ref{Xm}), we obtain (\ref{ave_rep_par}). We then introduce the R-S assumption, stating that all the considering overlaps are independent of the replica indices, reading
\setcounter{equation}{34}
\begin{align}
\mathbf Q^{(a)} &= \mathbf Q,  \mathbf T^{(ab)} = \mathbf T, \mathbf M^{(a)} = \mathbf M,\nonumber\\
\hat{\mathbf Q}^{(a)} &= \hat{\mathbf Q}, \hat{\mathbf T}^{(ab)} = \hat{\mathbf T}, \hat{\mathbf M}^{(a)} = \hat{\mathbf M}.\label{rs}
\end{align}
The R-S assumptions simply the expression of average replicated function. By combining the R-S equations (\ref{rs}) with Hubbard-Stratonovich (H-S) transform, we further obtain
\begin{align}
&\exp\left({\rm Tr}\left(\sum\limits_{b,a<b}^{n(n-1)/2}\hat{\mathbf T}^{(ab)T}\hat{\mathbf x}^{(a)T}\hat{\mathbf x}^{(b)}\right)\right)\nonumber\\
 = &\int D\mathbf z\exp\left({\rm Tr}\left(\sum\limits_{a = 1}^n\hat{\mathbf T}^{\frac{1}{2}}\sqrt{2}\mathfrak R(\mathbf z^T)\hat{\mathbf x}^{(a)}\right.\right.\nonumber\\
 &\left.\quad\quad\quad\quad\quad\quad\quad-\frac{1}{2}\sum\limits_{a = 1}^n\hat{\mathbf T}\hat{\mathbf x}^{(a)T}\hat{\mathbf x}^{(a)}\right)\nonumber
\end{align}
On the other hand, by adopting (\ref{rs}), the covariance matrix $\mathbf G$ can be further calculated as
\begin{align}
\mathbf G =& \boldsymbol\amalg_n\otimes\left(\frac{K}{M}\left(\mathbf C+\mathbf T-\mathbf M-\mathbf M^T\right)+\boldsymbol\Delta\right)\nonumber\\
\quad+&\mathbf I_n\otimes \frac{K}{M}\left(\mathbf Q-\mathbf T\right),
\end{align}
where $\boldsymbol\amalg_n$ stands for the $n\times n$ matrix with elements all equal to one. Accordingly, the eigenvalue set of $\mathbf G$ consists of the $2^L$ eigenvalues of matrix $\frac{K}{M}\left(\mathbf Q-\mathbf T\right)+n\left(\frac{K}{M}\left(\mathbf C+\mathbf T-2\mathbf M\right)+\boldsymbol\Delta\right)$, and $(n-1)$ subsets of $2^L$ eigenvalues with each subset consisting of the eigenvalues of $\frac{K}{M}\left(\mathbf Q-\mathbf T\right)$. Accordingly, we obtain (\ref{xm2}). By further considering the approximation $\lim\limits_{n\to 0}\mathbb E_{\mathbf x}\int D\mathbf z f(\mathbf x, \mathbf z)^n\approx \exp\left(n\mathbb E_{\mathbf x}\int D\mathbf z\log f(\mathbf x, \mathbf z)\right)$, the $\mathbb E_{\mathbf S,\mathbf Y}(\mathcal Z(\mathbf Y, \mathbf S)^n)$ can be reformulated as
\setcounter{equation}{37}
\begin{align}
\mathbb E_{\mathbf S,\mathbf Y}&(\mathcal Z(\mathbf Y, \mathbf S)^n) = \int{\rm d}\mathbf Q{\rm d}\hat{\mathbf Q}{\rm d}\mathbf M{\rm d}\hat{\mathbf M}{\rm d}\mathbf T{\rm d}\hat{\mathbf T}\nonumber\\
&\exp\left(Kn\tilde{\Phi}\left(\mathbf M, \mathbf Q, \mathbf T, \hat{\mathbf M}, \hat{\mathbf Q}, \hat{\mathbf T}\right)\right),\label{in}
\end{align}
where the involved function $\tilde{\Phi}$ can be formulated as (\ref{tid_phi}). To calculate the free entropy (\ref{rep_tri}), in the limit of $K$ and $n$, the integral in (\ref{in}) can be estimated by saddle point method, which is performed by taking the extremum of the potential (\ref{tid_phi}) with respect to the free parameters. Performing derivatives with respect to $\mathbf M$, $\mathbf T$ and $(\mathbf Q-\mathbf T)$, and setting the derivatives into zero, the extremum values satisfy
\setcounter{equation}{39}
\begin{align}
&\hat{\mathbf M}^{\star T}= \hat{\mathbf Q}^{\star T}+\hat{\mathbf T}^{\star T} = 2\left(\boldsymbol\Delta+\alpha^{-1}\left(\mathbf Q^{\star}-\mathbf T^{\star}\right)\right)^{-1}\nonumber\\
&\hat{\mathbf T}^{\star T} = 2\left(\boldsymbol\Delta+\alpha^{-1}\left(\mathbf Q^{\star}-\mathbf T^{\star}\right)\right)^{-1}\nonumber\\
&\quad\quad\times\left(\alpha^{-1}\left(\mathbf C+\mathbf T^{\star}-\mathbf M^{\star}-\mathbf M^{\star T}\right)+\boldsymbol\Delta\right)\nonumber\\
&\quad\quad\times\left(\boldsymbol\Delta+\alpha^{-1}\left(\mathbf Q^{\star}-\mathbf T^{\star}\right)\right)^{-1}.\label{extr}
\end{align}
We consider the Bayes-optimal case, where the distribution $p(\hat{\mathbf x})$ matches the true priori $p(\mathbf x)$, and we thus have
$\mathbf T ^{\star}= \mathbf M^{\star} = \mathbf M^{\star T}, \mathbf Q^{\star} = \mathbf C$. Since only the extrema of the potential (\ref{tid_phi}) that influence the estimate of saddle point method, we can substituting these conditions into (\ref{tid_phi}). Then, we can derive the functional form of the free entropy function in (\ref{fre_ent_E}) by setting $\mathbf E = \mathbf C-\mathbf M$. As a consequence, the \emph{Theorem 1} holds.

\setcounter{TempEqCnt}{\value{equation}}
\setcounter{equation}{40}
\begin{figure*}[tbp]
\begin{align}
&\left(\boldsymbol\Delta+\alpha^{-1}\mathbf E\right)^{-1} = \left(\sigma^2+\alpha^{-1}(a-b)\right)^{-1}\mathbf I_{2^L}-\alpha^{-1}b(\sigma^2+\alpha^{-1}(a-b))^{-1}\nonumber\\
&\quad\quad\quad\times(\sigma^2+\alpha^{-1}(a-b)+2^L\alpha^{-1}b)^{-1}\boldsymbol\amalg_{2^L},\label{woodburry1}\\
&\left(\boldsymbol\Delta+\alpha^{-1}\mathbf E\right)^{-\frac{1}{2}} = \left(\sigma^2+\alpha^{-1}(a-b)\right)^{-\frac{1}{2}}\mathbf I_{2^L}-\alpha^{-1}b\left(\sigma^2+\alpha^{-1}(a-b)+2^L\alpha^{-1}b\right)^{-\frac{1}{2}}\nonumber\\
&\quad\quad\times \left(\sigma^2+\alpha^{-1}(a-b)\right)^{-\frac{1}{2}}\left(\left(\sigma^2+\alpha^{-1}(a-b)\right)^{\frac{1}{2}}
+\left(\sigma^2+\alpha^{-1}(a-b)+2^L\alpha^{-1}b\right)^{\frac{1}{2}}\right)^{-1}\boldsymbol\amalg_{2^L}.\label{woodburry2}
\end{align}
\begin{align}
\Phi_1(a-b,b)& = 2^L\left(\frac{b(\sigma^2+\alpha^{-1}c)}{(\sigma^2+\alpha^{-1}(a-b)+2^L\alpha^{-1}b)(\sigma^2+\alpha^{-1}(a-b))}-\frac{\alpha\sigma^2+c}{\sigma^2+\alpha^{-1}(a-b)}\right)\nonumber\\
&-\alpha\log\left(\left(\sigma^2+\alpha^{-1}(a-b)+2^L\alpha^{-1}b\right)\left(\sigma^2+\alpha^{-1}(a-b)\right)^{2^L-1}\right)+\int D\mathbf z\log\Bigg(\sum\limits_{i = 1}^{2^L}\frac{1}{2^L}\nonumber\\
&\times\exp\bigg(-\frac{1}{\sigma^2+\alpha^{-1}(a-b)}-\frac{\alpha^{-1}b}{(\sigma^2+\alpha^{-1}(a-b)+2^L\alpha^{-1}b)(\sigma^2+\alpha^{-1}(a-b))}\nonumber\\
&+\frac{\sqrt{2}z_i}{(\sigma^2+\alpha^{-1}(a-b))^{\frac{1}{2}}}\bigg)+\frac{1}{2^L}\exp\bigg(-\frac{\alpha^{-1}b}{(\sigma^2+\alpha^{-1}(a-b)+2^L\alpha^{-1}b)(\sigma^2+\alpha^{-1}(a-b))}\nonumber\\
&+\frac{1}{\sigma^2+\alpha^{-1}(a-b)}+\frac{\sqrt{2}z_1}{(\sigma^2+\alpha^{-1}(a-b))^{\frac{1}{2}}}\bigg)\Bigg).\label{phi1}
\end{align}
\hrulefill
\end{figure*}
\setcounter{equation}{\value{TempEqCnt}}
\subsection{Proof of Proposition 1}\label{ap2}
Based on (\ref{symm}), we obtain the derivations using the Woodbury formula, which is given by (\ref{woodburry1})-(\ref{woodburry2}). In consideration of (\ref{symm}), (\ref{woodburry1}) and (\ref{woodburry2}), the function $\bar{\Phi}$ can be directly reformulated as a function of argument $(a-b)$ and $b$, given by (\ref{phi1}). We can observe from the equation (\ref{phi1}) that the parameter $b$ is independent of the integrand by defining the argument $(a-b)$. Accordingly, it is straightforward to take the derivative with respect to $b$ into zero. By combining the fact $c = 1/2^L$ that is inferred from the considered URA scenario setting, we obtain that the extreme point of argument $b$ satisfies
\setcounter{equation}{43}
\begin{equation}
b^{\star} = \frac{2^Lc-2^L(a-b)-1}{2^{2L}} = -\frac{a-b}{2^{L}},\label{b_star}
\end{equation}
for any parameter value of $(a-b)$ that satisfies its definition, i.e., $a-b\ge 0$. By plugging (\ref{b_star}) into (\ref{phi1}) with the denotation $d\triangleq a-b$, we can directly obtain a free entropy function with single scalar argument $d$, which is given by (\ref{phi2}).

\subsection{Proof of Proposition 2}\label{ap3}
Reviewing the equation (\ref{tid_phi}) and (\ref{extr}), the fixed point condition with respect to parameter $\hat{\mathbf M}$ is
\begin{align}
\mathbf M^{\star} = &\int D\mathbf z\sum\limits_{\mathbf x}p(\mathbf x)\nonumber\\
\times&\mathbf x^T \boldsymbol\eta\left(\left(\mathbf x+\mathbf z(\boldsymbol\Delta+\frac{1}{\alpha}\mathbf E^{\star})^{\frac{1}{2}}\right), \boldsymbol\Delta+\frac{1}{\alpha}\mathbf E^{\star}\right).
\end{align}
The function (\ref{eta}) can be interpreted as the MMSE denoiser over system (\ref{rmod}) with the knowledge of (\ref{con_pri}) and a particular element of (\ref{eta}) can be interpreted as the posterior probability that the component in this position is non-zero. We then have
\begin{align}
\mathbf E^{\star} = &\mathbf C-\int D\mathbf z\sum\limits_{\mathbf x}p(\mathbf x)\mathbf x^T\boldsymbol\eta(\mathbf r, \boldsymbol\Sigma)\nonumber\\
= &\int D\mathbf z\sum\limits_{\mathbf x}p(\mathbf x)\left(\mathbf x-\boldsymbol\eta(\mathbf r, \boldsymbol\Sigma)\right)^T\left(\mathbf x-\boldsymbol\eta(\mathbf r, \boldsymbol\Sigma)\right).\label{SE}
\end{align}
We can observe from (\ref{SE}) that $\mathbf E^{\star}$ can be regarded as the error matrix of signal model (\ref{rmod}). Following the statements in \cite{Tanaka2002, Barbier2017}, we then have that the signal model (\ref{equi_mod}) in each sub-block of second phase is statistically equivalent to (\ref{rmod}). Accordingly, based on the definition (\ref{def_Pe}), the probability $P_e$ with MAP criterion in the $K\to\infty$ limit can be formulated based on the matrix $\mathbf E^{\star}$ via (\ref{exp_pe}).

\bibliography{a}

\end{document}